\documentclass{article}

\newcommand{\ifFull}[1]{#1}
\newcommand{\ifConference}[1]{}
\newcommand{\ifSubmission}[1]{}
\PassOptionsToPackage{numbers,sort}{natbib}

\ifSubmission{\usepackage{neurips_2024}}
\ifFull{\usepackage[final]{neurips_2024}}
\ifConference{\usepackage[final]{neurips_2024}}

\usepackage{amsmath,mathtools} 
\usepackage{amssymb} 
\usepackage{bm} 
\usepackage{stmaryrd} 
\usepackage{amsthm}

\usepackage[utf8]{inputenc} 
\usepackage[T1]{fontenc}    
\usepackage[colorlinks=true,citecolor=purple,linkcolor=purple,hypertexnames=false]{hyperref}       
\usepackage{url}            
\usepackage{booktabs}       
\usepackage{amsfonts}       
\usepackage{nicefrac}       
\usepackage{microtype}      
\usepackage{xcolor}         
\usepackage[capitalise,nameinlink]{cleveref}
\usepackage{todonotes}
\usepackage{newfile}

\usepackage[shortlabels]{enumitem}

\newtheorem{theorem}{Theorem}
\newtheorem{proposition}[theorem]{Proposition}
\newtheorem{lemma}[theorem]{Lemma}
\newtheorem{remark}[theorem]{Remark}
\newtheorem{example}[theorem]{Example}

\newcommand{\ltltl}{\ensuremath{(\text{LT})^2\text{L}}}

\newcommand{\Q}{\mathbb{Q}}
\newcommand{\OMIT}[1]{}

\newcommand{\ialphabet}{\Sigma}

\newcommand{\TCzero}{\text{TC}^0}
\newcommand{\ACzero}{\text{AC}^0}

\newcommand{\automatonfont}[1]{\mathfrak{#1}}
\newcommand{\monoidfont}[1]{\mathbb{#1}}
\newcommand{\structurefont}[1]{\mathcal{#1}}

\makeatletter
\newcommand{\@abbrev}[3]{
  \def\c@a@def##1{
    \if ##1.
    \relax
    \else
    \@ifdefinable{\@nameuse{#1##1}}{\@namedef{#1##1}{#2##1}}
    \expandafter\c@a@def
    \fi
  }
  \c@a@def #3.
}

\@abbrev{bb}{\mathbb}{ABCDEFGHIJKLMNOPQRSTUVWXYZ}
\@abbrev{bf}{\mathbf}{ABCDEFGHIJKLMNOPQRSTUVWXYZabcdefghijklmnopqrstuvwxyz}
\@abbrev{bit}{\boldsymbol}{ABCDEFGHIJKLMNOPQRSTUVWXYZabcdefghijklmnopqrstuvwxyz}
\@abbrev{cal}{\mathcal}{ABCDEFGHIJKLMNOPQRSTUVWXYZ}
\@abbrev{frak}{\mathfrak}{ABCDEFGHIJKLMNOPQRSTUVWXYZabcdefghijklmnopqrstuvwxyz}
\@abbrev{up}{\mathrm}{ABCDEFGHIJKLMNOPQRSTUVWXYZabcdefghijklmnopqrstuvwxyz}
\@abbrev{scr}{\mathscr}{ABCDEFGHIJKLMNOPQRSTUVWXYZ}
\@abbrev{sf}{\mathsf}{ABCDEFGHIJKLMNOPQRSTUVWXYZabcdefghijklmnopqrstuvwxyz}
\@abbrev{ov}{\ov}{ABCDEFGHIJKLMNOPQRSTUVWXYZabcdefghijklmnopqrstuvwxyz}
\@abbrev{struct}{\structurefont}{ABCDEFGHIJKLMNOPQRSTUVWXYZ}
\@abbrev{aut}{\automatonfont}{ABCDEFGHIJKLMNOPQRSTUVWXYZ}
\@abbrev{mono}{\monoidfont}{ABCDEFGHIJKLMNOPQRSTUVWXYZ}
\@abbrev{graph}{\structurefont}{ABCDEFGHIJKLMNOPQRSTUVWXYZ}
\makeatother

\newcommand{\bc}{\bm{c}}
\newcommand{\bd}{\bm{d}}
\newcommand{\be}{\bm{e}}
\newcommand{\ba}{\bm{a}}
\newcommand{\bb}{\bm{b}}
\newcommand{\br}{\bm{r}}

\newcommand{\bs}{\bm{s}}
\newcommand{\bt}{\bm{t}}
\newcommand{\bu}{\bm{u}}
\newcommand{\bv}{\bm{v}}
\newcommand{\bw}{\bm{w}}
\newcommand{\bx}{\bm{x}}
\newcommand{\by}{\bm{y}}
\newcommand{\bz}{\bm{z}}
\newcommand{\N}{\mathbb{N}}
\newcommand{\R}{\mathbb{R}}
\newcommand{\Z}{\mathbb{Z}}

\title{The Power of Hard Attention Transformers on Data Sequences:
A Formal Language Theoretic Perspective}

\author{%
  Pascal Bergsträßer\\
  RPTU Kaiserslautern-Landau\\
  67663 Kaiserslautern, Germany\\
  \texttt{bergstraesser@cs.uni-kl.de}
  \And
  Chris Köcher\\
  MPI-SWS\\
  67663 Kaiserslautern, Germany\\
  \texttt{ckoecher@mpi-sws.org}
  \AND
  Anthony Widjaja Lin\\
  MPI-SWS and RPTU Kaiserslautern-Landau\\
  67663 Kaiserslautern, Germany\\
  \texttt{awlin@mpi-sws.org}
  \And
  Georg Zetzsche\\
  MPI-SWS\\
  67663 Kaiserslautern, Germany\\
  \texttt{georg@mpi-sws.org}
}

\begin{document}

\maketitle

\begin{abstract}
    Formal language theory has recently been successfully employed to unravel 
    the power of transformer encoders. This setting is primarily applicable in
    Natural Language Processing (NLP), as a token embedding function (where
    a bounded number of tokens is admitted) is first applied before feeding 
    the input to the transformer. 
    On certain kinds of data (e.g. time 
    series), we want our transformers to be able to handle \emph{arbitrary} 
    input sequences of numbers (or tuples thereof) without \emph{a priori} 
    limiting the values of these numbers. In this
    paper, we initiate the study of the expressive power of transformer encoders
    on sequences of data (i.e. tuples of numbers).
    Our results indicate an increase in expressive power of 
    hard attention transformers over data sequences, in stark contrast to the 
    case of strings.
    In particular, we prove that Unique Hard Attention Transformers (UHAT) over 
    inputs as data sequences no longer lie within the circuit complexity 
    class AC0 (even without positional encodings), unlike the case of string
    inputs, 
    but are still within the complexity class TC0 (even with positional
    encodings). Over strings, UHAT without positional encodings capture only
    regular languages. In contrast, we show that over data sequences
    UHAT can capture non-regular properties.
    Finally, we show that UHAT capture languages 
    definable in an extension of linear temporal logic with unary numeric
    predicates and arithmetics. 
\end{abstract}

\section{Introduction}
\label{sec:intro}

Recent years have witnessed the success of transformers \citep{Vaswani} 
in different applications, including natural language processing
\citep{BERT}, computer vision \citep{vision-transformers}, speech recognition
\citep{speech-transformers}, and time series analysis
\citep{time-series-survey,Zhou21}. In the quest to better understand the ability
and limitation of transformers, theoretical investigations have actively been 
undertaken in the last few years. 
Among others, \emph{formal language theory} has been successfully applied to reveal
deep insights into the expressive power of transformers, 
e.g., see the recent survey \citep{transformers-survey} and 
\citep{ACY24,barcelo2024logical,Perez21,Strobl23,saturated,Hao22,transformers-transducers,Chiang23,dusell24,Hahn20}.
In particular, relevant questions pertain to the power of various attention 
mechanisms, bounded/unbounded precision, positional encoding functions, and 
interplay between encoders and decoders, among many others. 

One common assumption in the formal language theoretic approach to transformers
is that the input sequence ranges over a \emph{finite} set $\ialphabet$ (called
alphabet), which is then to be fed into a transformer after applying a
%
token embedding function of the form $f: \ialphabet \to \R^d$. As a by-product,
the number of tokens is finite.
In certain applications 
(e.g. time series forecasting \citep{iTransformer-iclr24}), 
we want our transformers to be able to handle \emph{arbitrary} 
input sequences of numbers (or tuples thereof) without \emph{a priori} 
limiting the values of these numbers. Moreover, numbers could be compared using
arithmetic and (in)equality, which is not the case for elements of alphabets
considered in formal language theory.
%
For this reason, we propose to investigate 
the expressive power of \emph{transformers over data sequences}, which 
takes us to the setting of formal language theory over the alphabet
$\ialphabet = \Q^d$, for some $d \in \Z_{> 0}$, e.g., see \citep{atom-book,DV21}.
That is, \emph{what properties of a sequence of (tuples of) numbers
can be recognized by transformers?}

\textbf{Connections to circuit complexity.} Existing work 
has revealed intimate connections between transformers and circuit complexity.
In particular, let us consider the following class of transformer encoders that 
has been the main focus of many recent papers: \emph{Unique Hard Attention 
Transformers (UHAT)}. Among others, UHAT allows arbitrary positional encoding
and an attention mechanism that picks a vector at a unique minimum position in
the sequence that maximizes the attention score.
It is known that the class of formal languages recognized by 
UHAT is a strict subset of the circuit complexity
class $\ACzero$ (cf. \citep{Hao22,Hahn20,barcelo2024logical}), i.e., each UHAT 
can be simulated by a family of boolean circuits of constant depth. More
concretely, this entails among others that UHAT cannot compute the parity
(even/oddness) of the number of occurrences of any given letter $a$
in the input string (for strings over an alphabet containing at least two 
letters). 

\textbf{Our first contribution} is that UHAT over data sequences (even without 
positional encodings) is \emph{no longer} contained in $\ACzero$, unlike the case of finite 
number of tokens. Instead, we show that UHAT can be captured by the circuit
complexity class $\TCzero$, which extends $\ACzero$ circuits with majority
gates. 
\begin{theorem}\label{main-result-tc0}
    UHAT with positional encoding over data sequences is in $\TCzero$ but not
in $\ACzero$.
\end{theorem}
This complexity upper bound allows us to deduce the expressive power of
UHAT over data sequences by using complexity theory.
For example, since UHAT accepts only $\TCzero$ languages, successfully
constructing a UHAT (e.g. through learning) for
        \[
            \mathsf{SQRTSUM} := \left\{(a_1,b_1),\ldots,(a_n,b_n) 
\,\middle|\,
        \sum_{i=1}^n \sqrt{a_i} \leq \sum_{i=1}^n \sqrt{b_i},
        \text{and each } (a_i,b_i) \in \Z_{>0}\times \Z_{> 0},
\right\},
    \]
would constitute a major breakthrough in complexity theory (cf.
\citep{ABKM06,GGJ76}), i.e.,
showing that $\mathsf{SQRTSUM}$ is in the complexity class 
$\TCzero\subseteq\text{P}/\text{poly} $. A byproduct of our proof is that for each length, the set of accepted sequences is a semialgebraic set. This implies, e.g., that the graph $\{(x,e^x) \mid x\in\R\}\subseteq\R^2$ of $x\mapsto e^x$ (viewed as a set of length-$1$ sequences) is not accepted by UHAT.

\textbf{Connection to regular languages over data sequences.} Recent results
have revealed surprising connections between regular languages and formal 
languages recognizable by transformer encoders. In particular, it was proven
(cf. \citep{ACY24}) that languages
recognizable by UHAT (even \emph{with no positional encodings}) form a strict subset of
regular languages, namely those that are ``star-free'' or, equivalently, 
definable in First-Order Logic (FO), or Linear Temporal Logic (LTL). With
positional encodings, similar connections 
hold, by extending the logics with \emph{unary numerical predicates} (cf.
\citep{ACY24,barcelo2024logical}).

To investigate whether such connections extend to data sequences, we bring forth
\emph{formal languages theory over infinite alphabets} (cf.,
\citep{atom-book,DV21}), which has been an active research field in the last
decade or so with applications to programming languages and databases (to name a
few), e.g., see \citep{DV21,popl12,LMV16}. \textbf{Our second contribution} is a 
language over data
sequences
recognizable by UHAT without positional encodings that lies beyond existing
formal models over infinite alphabets (in particular, ``regular'' ones). This
shows the strength of UHAT over data sequences even without positional
encodings, in stark contrast to the case of finite alphabets.
\begin{theorem}\label{thm:non-regular}
    There is a non-regular language over $\ialphabet = \Q^d$ that is accepted
    by masked UHAT with no positional encoding.
\end{theorem}

Finally, to better understand languages over data sequences recognizable by 
UHAT, \textbf{our third contribution} is to show how UHAT can recognize 
languages definable by the so-called
\emph{Locally Testable LTL} ($\ltltl$), which extends LTL with unary numerical
predicates and \emph{local arithmetic tests} for fixed-size windows over the
input sequence. Using $\ltltl$, we can see \emph{at a glance} what can be
expressed by UHAT over data sequences. For one, this includes the well-known 
\emph{Simple Moving Averages}. As another example, using $\ltltl$ it can be 
easily shown that UHAT
can capture linear recurrence sequences considered in the famous Skolem 
problem and 
discrete linear dynamical systems \citep{luca23,lipton22,skolem-survey}, i.e., 
sequences of the form
    $\bm x, A\bm x, \dots, A^n\bm x$
such that $n \ge 0$ is minimal with $\bm y A^n \bm x = 0$ 
where $\bm y \in \Q^{1 \times d}$ and $A \in \Q^{d \times d}$ are fixed and $\bm x \in \Q^d$.
\begin{theorem}\label{thm:ltl}
Every \ltltl-definable language is accepted by UHAT with positional encoding.
\end{theorem}

\OMIT{
In particular, this would allow
understanding the expressive power of transformers, when the input sequence is a
sequence of \emph{data} (in particular, numbers), which are relevant in man
}
\OMIT{
As previously mentioned, transformers have been successfully applied not only to
textual data, but also to other types of data, including visual data, audio
data, time series, and more generally sequences of (tuples of) numbers. For this
reason, we would like to understand what can be learned by transformers on 
such forms of data; in particular, \emph{are the recent insights through formal
language theory valid here?} To answer this question, we observe that

This is a by-product that the \emph{token embedding function} allowed in
the formal language theoretic approach to transformers is of the form $f: 
\ialphabet \to \R^d$, where $\ialphabet$

we note that one underlying assumption in the formal
language theoretic approach to transformers is the

recall the standard assumption in formal language 
theory that the elements allowed in the input sequence range over a 
\emph{finite} set $\ialphabet$ (called alphabet). 
In turn, 
the \emph{token embedding function} allowed in the formal language 
theoretic approach to transformers is of the form
$f: \ialphabet \to \R^d$,

which means that
only a \emph{finite} number of tokens is allowed. 
%
}







\textbf{Technical challenges.} Obtaining our results poses several challenges.
First, for the $\TCzero$ upper bound, we need to use Boolean (constant depth)
circuits to simulate UHATs, in which real constants can occur (in affine
transformations or positional encodings).  While in $\TCzero$, it is known that
majority gates can be used to perform multiplication of
rationals~\cite{chandra1984constant}, arithmetic with reals requires infinite
precision and cannot be done with Boolean circuits. To this end, we compute
rational approximations of reals accurate enough to preserve the
acceptance condition for inputs up to a particular length $n$.

Here, a naive attempt would be to replace each real occurring in the UHAT in
affine transformations and the positional encoding by some rational
approximation.  However, this is not possible, meaning \emph{any} rational
approximation would change the behavior on input sequences of length $n=3$,
even for low-dimensional vectors with entries in $\{0,1\}$. Indeed, there is a
UHAT involving real numbers $\alpha,\beta$ that accepts a simple sequence of
$\{0,1\}$-vectors if and only if $\alpha\beta=2$ and $\alpha=\beta$, i.e.\
$\alpha=\beta=\sqrt{2}$. Thus, $\alpha$ and $\beta$ cannot be replaced by
rationals, even for very short inputs (see
\cref{appendix-example-real-parameters} for details).

Instead, we show that a UHAT can be translated into a small Boolean combination
of polynomial inequalities. This format has the advantage that---as we show
using convex geometry---the real coefficients of those polynomials
\emph{can} be replaced by suitably chosen rational numbers. In turn, the
layer-by-layer construction of these polynomial inequalities requires a
carefully chosen data structure to encode the function computed by a sequence
of transformer layers. For example, we show that the resulting Boolean
combinations of polynomial inequalities have a bound on the number of
alternations between conjunctions and disjunctions, which is crucial for
constructing $\TCzero$-circuits.

Another key challenge occurs in the translation from \ltltl\ to UHATs: In the
inductive construction, we need to represent truth values using reals in
$[0,1]$. To implement negation, we use a UHAT gadget (with positional encodings)
that can normalize these
truth values to $\{0,1\}$.

\textbf{Notation.} In the sequel, we assume some background from
computational complexity, in particular circuit complexity (see the book
\cite{Vollmer-book}). In particular, we use the circuit complexity class 
$\ACzero$, which defines a class of problems that are computable by a nonuniform
family of constant-depth boolean circuits, where each gate permits an unbounded
fan-in (i.e. arbitrary many inputs). Similarly, the complexity class $\TCzero$
is an extension of $\ACzero$, where majority gates are allowed. It is well-known
that $\ACzero \subsetneq \TCzero$. Assuming \emph{uniformity}, both $\ACzero$
and $\TCzero$ are contained in the class of problems solvable in
polynomial-time. For \emph{nonuniformity}, these classes are contained in the
complexity class $\text{P}/\text{poly}$, which admits (nonuniform) 
polynomial-size circuits. It is a long-standing open problem whether numerical
analysis (e.g. square-root-sum) is in $\text{P}/\text{poly}$, e.g., see
\cite{ABKM06}.

\section{Transformer encoders and their languages}
\label{sec:prelims}

In the following, we adapt the setting of formal language theory
(see \cite{transformers-survey,Hao22,barcelo2024logical}) to data sequences.
For a vector $\bm a = (a_1,\dots,a_d)$ we write $\bm a[i,j] := (a_i,\dots,a_j)$ for all $1\le i \le j \le d$
and if $i = j$, we simply write $\bm a[i]$.
For a set $S$ we denote the set of (potentially empty) sequences of elements from $S$ by $S^*$.
We write $S^+$ for the restriction to non-empty sequences. 
We consider languages $L$ over the infinite alphabet $\Sigma = \Q^d$, for some
integer $d > 0$. That is, $L$ is a set of sequences of $d$-tuples over rational
numbers. We will define a UHAT (similarly as in previous papers that study formal
language theoretic perspectives) as a length preserving map $(\Q^d)^* \to (\R^e)^*$.

\textbf{Standard encoder layer with unique hard attention.}
A standard encoder layer is defined by three affine transformations $A,B \colon \R^d \to \R^d$ and $C \colon \R^{2d} \to \R^e$.
For a sequence $\bm v_1,\dots,\bm v_n \in \R^d$ with $n \ge 1$ we define the \emph{attention vector} at position $i \in [1,n]$ as
$\bm a_i := \bm v_j$
with $j \in [1,n]$ minimal such that the \emph{attention score} 
$\langle A\bm v_i, B\bm v_j \rangle$ is maximized.
The layer outputs the sequence $C(\bm v_1,\bm a_1),\dots,C(\bm v_n,\bm a_n)$.

\textbf{ReLU encoder layer.}
A ReLU layer for some $k \in [1,d]$ on input $\bm v_1,\dots,\bm v_n \in \R^d$ 
applies the ReLU function to the $k$-th coordinate of each $\bm v_i$, i.e.
it outputs the sequence $\bm v'_1,\dots,\bm v'_n$ where $\bm v'_i := (\bm
v_i[1,k-1],\max\{0,\bm v_i[k]\},\bm v_i[k+1,n])$. [Equivalently, one could
instead allow a feed-forward network at the end of an encoder layer (see
\cite{transformers-survey,Hao22,barcelo2024logical}).]

\textbf{Transformer encoder.}
A \emph{unique hard attention transformer encoder} (UHAT) is a repeated application of 
standard encoder layers with unique hard attention and ReLU encoder layers.
Clearly, using an alternation of standard layers and ReLU layers, we can assume that the output of a UHAT layer 
is an arbitrary composition of affine transformations and component-wise ReLU application.
In particular, these compositions may use the functions $\max$ and $\min$.

\textbf{Languages accepted by UHATs.}
The notion of ``languages'' accepted by a UHAT (i.e. a set of accepted
sequences) can be defined, depending on whether a \emph{positional encoding} is
permitted. 
If it is permitted,
a language $L \subseteq (\Q^d)^+$ is accepted by a UHAT $T$ if and only if
there exists a positional encoding function $p \colon \N \times \N \to \R^s$ and an acceptance vector $\bm t \in \R^e$ such that
on every sequence 
\begin{equation}\label{eq:input}
(p(1,n+1),\bm v_1),\dots,(p(n,n+1),\bm v_n), (p(n+1,n+1),\bm 0)
\end{equation}
$T$ outputs a sequence $\bm v'_1,\dots,\bm v'_{n+1} \in \R^e$ with
$\langle \bm t,\bm v'_1 \rangle > 0$ if and only if $(\bm v_1,\dots,\bm v_n) \in L$.
Note that if $T_1$ and $T_2$ are UHATs with positional encoding that realize functions
$f_1 \colon (\Q^d)^* \to (\Q^e)^*$ and $f_2 \colon (\Q^e)^* \to (\Q^r)^*$, 
then there is a UHAT $T_2 \circ T_1$ with positional encoding that realizes the composition $f_2 \circ f_1$
by using a positional encoding that combines the positional encodings of $T_1,T_2$.

In the above definition we appended an additional zero vector to the end of the input.
Over finite alphabets it is often assumed that the input sequence is extended
with a special unique end-of-input marker (e.g. see
\cite{transformers-survey,Hao22}).
When the input is a sequence of (tuples of) numbers, if we allow positional 
encoding, then the zero vector at the end of the input can be turned into a 
unique vector marking the end of the input (see \cref{sec:ltl}).
Without positional encoding, however, we have to explicitly make the zero vector at the end of the input unique.
That is, a UHAT without positional encoding is initialized with the sequence
$(1,\bm v_1),\dots,(1,\bm v_n),\bm 0 \in \Q^{d+1}$ instead; this ensures, among
others, that the end-of-input marker does not appear in the actual input.

In the definition of a standard encoder layer the attention vector at position $i \in [1,n]$ can be any vector in the sequence
$\bm v_1,\dots,\bm v_n$.
Using \emph{masking}, one can restrict the attention vector to vectors of certain positions.
A UHAT with \emph{past masking} restricts the attention vector $\bm a_i$ at position $1 \le i < n$ 
to be contained in the subsequence $\bm v_{i+1},\dots,\bm v_n$ and at position $n$ to $\bm a_n := \bm v_n$.

\section{UHAT and $\TCzero$}\label{sec:tc0}
In this section, we prove \cref{main-result-tc0}. First, we show that all languages of UHAT (even with positional encoding) belong to the class $\TCzero$. Then, we show that there is a UHAT (even without positional encoding) whose language is $\TCzero$-hard under $\ACzero$-reductions. We begin with the proof that all UHAT languages belong to $\TCzero$.

\paragraph{Input encoding} A language $L\subseteq\Sigma^*$ over a finite
alphabet $\Sigma$ belongs to $\TCzero$ if for every input length $n$, there
is a circuit of size polynomial $n$, such that the circuit consists
of input gates (for each input position $i$, and each letter $a\in\Sigma$,
there is a gate that evaluates to ``true'' if and only if position $i$ of the
input words carries an $a$), Boolean gates (computing the AND, OR, or NOT
function) and majority gates (evaluating to true if more than half of their
input wires carry true). Here, AND, OR, and majority gates can
have arbitrary fan-in. In order to define what it means that a language
$L\subseteq (\Q^d)^+$ belongs to $\TCzero$, we need to specify an encoding as
finite-alphabet words. To this end, we encode a sequence $\bm
u_1,\ldots,\bm{u}_n$ with $\bm u_i\in\Q^d$ as a string $v_1\#\cdots\#v_n$,
where $v_i=p_1/q_1\square\cdots\square p_d/q_d$ with $p_j,q_j\in\{\mathord{-},0,1\}^*$. Here,
$v_i\in\{\mathord{-},\mathord{/},\square,0,1\}^*$ represents the vector $\bm u_i\in\Q^d$ such that $\bm
u_i[j]=\tfrac{a_{j}}{b_{j}}$, $a_j,b_j\in\Z$, and $p_j,q_j\in\{\mathord{-},0,1\}^*$ are the binary expansions
of $a_j$ and $b_j$.

\begin{remark}
The main challenge in proving \cref{main-result-tc0} is that the constants
appearing in a UHAT can be real numbers. These can in general not be avoided:
There are UHAT languages over $\Sigma=\Q$ (even without positional encoding)
that cannot be accepted by UHAT with rational constants (even with positional
encoding). For example, for every real number $r>0$, one can\footnote{First,
use a standard encoding layer to transform the sequence
$(x_1,\ldots,x_n)\in\Sigma^n$ into
$(r^{-1}x_1-1,\ldots,r^{-1}x_n-1)\in\Sigma^n$. Then, accept if and only if the
left-most vector is $>0$.} construct a UHAT for the language $L_r$ of all
sequences over $\Sigma=\Q$ where the first letter is $>r$. Note that $L_r\ne L_{s}$
for any $r,s>0$, $r\ne s$.  However, there are clearly only countably many
languages over $\Sigma=\Q$ accepted by UHAT with rational constants where
membership only depends on the first letter\footnote{This is because all
employed affine transformations (which are rational matrices), but also the
values $p(1,2), p(2,2)\in \Q$ of the positional encoding can be described using
finitely many bits.}. Thus, there are uncountably many real $r>0$ such that
$L_r$ is not accepted by a UHAT with rational constants.
\end{remark}

The construction of $\TCzero$ circuits comprises three steps. In Step~I, we show that the set of accepted length-$n$ sequences can be represented by a Boolean combination of polynomial inequalities. Importantly, (i)~this representation, called ``polynomial constraints'' is polynomial-sized in $n$, and (ii)~the number of alternations between conjunction and disjunction is bounded (i.e.~independent of $n$). The polynomials in this representations can still contain real coefficients. 
In Step~II, we show that if we restrict the input not only to length-$n$ sequences, but to rational numbers of size $\le m$, then we can replace all real coefficients of our polynomials by rationals of size polynomial in $m$ and $n$, without changing the language (among vectors of size $\le m$). 
%
In Step~III, we implement a $\TCzero$ circuit. Here, it is important that the number of alternations between conjunctions and disjunctions in our polynomial constraints is bounded, because the depth of the circuit is proportional to this number of alternations.

\newcommand{\preplong}{polynomial UHAT representation} 
\newcommand{\prep}{PUR} 
\newcommand{\setof}[1]{\llbracket #1\rrbracket} 
\newcommand{\normof}[1]{\|#1\|_2} 

\paragraph{Step I: UHAT as polynomials}
We first consider a formalism to describe a set of sequences over $\Q^d$. We consider such sequences $(\bx_1,\ldots,\bx_n)$ of length $n$, where $\bx_i\in\Q^d$ for each $i$. In this case, we also abbreviate $\bar{\bx}=(\bx_1,\ldots,\bx_n)$. 
A \emph{polynomial constraint} (PC) is a positive Boolean combination (i.e., without negation) of constraints
of the form $p(\bar{\bx})>0$ or $p(\bar{\bx})\ge 0$, where
$p\in\R[X_1,\ldots,X_{d\cdot n}]$ is a polynomial with real coefficients. Here,
plugging $\bar{\bx}\in (\Q^d)^n$ into $p$ is defined by assigning the $d\cdot
n$ rational numbers in $\bar{\bx}$ to the $d\cdot n$ variables
$X_1,\ldots,X_{d\cdot n}$. The PC $\alpha$ \emph{accepts} a sequence of vectors $\bar{\bx}\in(\Q^d)^n$, if the
Boolean formula evaluates to true when plugging $\bar{\bx}$ into the polynomials $p$ in $\alpha$.
The set of accepted sequences is denoted by $\setof{\alpha}$.
Now let $a\in\N$. A PC
has \emph{$a$ alternations} if the positive Boolean combination has
$a$ alternations between disjunctions and conjunctions.

In the following, a \emph{constrained polynomial representation (CPR)} can be used to
compute from a sequence of inputs $(\bx_1,\dots,\bx_n)$ with $\bx_1,\dots,\bx_n\in\R^{d'}$ a new
sequence of outputs $(\by_1,\dots,\by_n)$ with $\by_1,\dots,\by_n\in\R^d$. Formally,
a CPR comprises for each $i\in\{1,\ldots,n\}$ a sequence
$(\varphi_1,D_1),\ldots,(\varphi_{s_i},D_{s_i})$ of pairs $(\varphi_j,D_j)$, where
each pair $(\varphi_j,D_j)$ is a ``conditional assignment'': each $(\varphi_j,D_j)$ tells us that if the condition $\varphi_j$ is satisfied, then we return $D_j(\bar{\bx})$. More precisely:
(i)~each $\varphi_j$ is a polynomial constraint where all polynomials have degree $\le 2$, 
(ii)~for any $j\ne m$, the constraints $\varphi_j$ and $\varphi_m$ are mutually exclusive, and
(iii)~each $D_j\colon \R^{d'\cdot n}\to\R^{d}$ is an affine transformation.
Because of their role as conditional assignments, we also write $\varphi_j\to D_j$ for such pairs.
For $a\in\N$, we say that the CPR is \emph{$a$-alternation-bounded} if each of the
formulas $\varphi_j$ has at most $a$ alternations. A CPR as above computes a
function $\R^{d'\cdot n}\to\R^{d\cdot n}$: Given
$\bar{\bx}=(\bx_1,\ldots,\bx_n)$ with $\bx_1,\ldots,\bx_n\in\R^{d'}$, it computes
the sequence $(\by_1,\ldots,\by_n)$ if for every $i\in\{1,\ldots,n\}$, we have
$\by_i=D_j(\bar{\bx})$, provided that $j$ is the (in case of existence uniquely
determined) index for which $\varphi_j(\bar{\bx})$ is satisfied. The \emph{size} of PCs and CPRs are their bit lengths (see \cref{appendix-tc0} for details).
\begin{proposition}\label{prop:uhatToPC}
	Fix a UHAT with positional encoding and $\ell$ layers. For any given sequence length $n$, 
	there exists a polynomial-sized PC $\alpha$ with $O(\ell)$ alternations such that $\setof{\alpha}$ equals the set of accepted sequences of length $n$.
\end{proposition}
Note that \cref{prop:uhatToPC} implies that the set of sequences of each length $n$ is a semialgebraic set~\cite{Mishra93}.
The proof is by induction on the number of layers, which requires a slight strengthening:
\begin{lemma}\label{construct-cpr}
	Fix a UHAT with positional encoding and $\ell$ layers. For any given sequence length $n$, one can construct in
	polynomial time an $O(\ell)$-alternation-bounded CPR computing the function $\R^{d\cdot (n+1)}\to\R^{e\cdot (n+1)}$
	computed by the UHAT.
\end{lemma}
\begin{proof}
  We prove the statement by induction on the number of layers. First, we consider the positional encoding $p\colon\N\times\N\to\R^d$ as some affine transformations $P_i\colon\R^{d\cdot (n+1)}\to\R^d$ mapping the input sequence $\bar{\bx}$ to $\bx_i+p(i,n+1)$. Then we obtain a CPR with $\top\to P_i$ for each $1\leq i\leq n+1$. Now, suppose the statement is shown for $\ell$ layers
	and consider a UHAT with $\ell+1$ layers. Suppose that the first $\ell$ layers of our UHAT compute a function $\R^{d'\cdot (n+1)}\to\R^{d\cdot (n+1)}$, and the last layer computes a function $\R^{d\cdot (n+1)}\to\R^{e\cdot (n+1)}$.
	By induction, we have a polynomial size CPR consisting of conditional assignments $\varphi_{i,k}\to D_{i,k}$ for every $i\in\{1,\ldots,n+1\}$ and $1\le k\le s_i$. Here, each $D_{i,j}$ is an affine transformation $\R^{d'\cdot (n+1)}\to\R^d$.

  Let us first consider the case that the last layer of our UHAT is a standard encoder layer.
	For each $(i,I,j,J)\in\{1,\ldots,n+1\}^4$, we build the conditional assignment using the formula $\psi_{i,I,j,J}$:
\begin{align*}
\bigwedge_{m=1}^{j-1} \left(
\bigvee_{M=1}^{s_m} \varphi_{m,M}\wedge p_{i,I,j,J,m,M}(\bar{\bx})>0
\right) \wedge
	\bigwedge_{m=j+1}^{n+1} \left(
\bigvee_{M=1}^{s_m} \varphi_{m,M}\wedge p_{i,I,j,J,m,M}(\bar{\bx})\ge 0
\right)
\end{align*}
where $p_{i,I,j,J,m,M}(\bar{\bx})$ is the polynomial $\langle AD_{i,I}\bar{\bx},~BD_{j,J}\bar{\bx}-BD_{m,M}\bar{\bx}\rangle$. Then, the conditional assignment is $\varphi_{i,I}\wedge\varphi_{j,J}\wedge\psi_{i,I,j,J}\to C(D_{i,I}\bar{\bx},D_{j,J}\bar{\bx})$. Here, the idea is that (i)~$\varphi_{i,I}$ expresses that the $I$-th conditional assignment was used to produce the $i$-th cell in the previous layer, (ii) $\varphi_{j,J}$ says the $J$-th conditional assignment was used to produce the $j$-th vector in the previous layer, and (iii)~$\psi_{i,I,j,J}$ says the vector $\bx_j$ yields the maximal attention score for the input $\bx_i$, meaning (iii-a) for all positions $m<j$, $\bx_m$ has a lower score than $\bx_j$ (left parenthesis), and (iii-b) for all positions $m>j$, $\bx_m$ has at most the score of $\bx_j$ (right parenthesis). In (iii-a) and (iii-b), we first find the index $M$ of the conditional assignment used to produce the $m$-th cell in the previous layer.
	Note that then indeed, all the PCs $\psi_{i,I,j,I}$ are mutually
	exclusive. Moreover, the polynomials $\langle
	AD_{i,I}\bar{\bx},~BD_{j,J}\bar{\bx}-BD_{m,M}\bar{\bx}\rangle$ have indeed degree $2$
	and are of size polynomial in $n$.
	Furthermore, if the assignments $\varphi_{i,k}\to D_{i,k}$ had at
	most $a$ alternations, then the new assignments have at most $a+3$
	alternations.
Finally, the case of ReLU layers is straightforward (see \cref{appendix-tc0}).
\end{proof}

Finally, the proof of \cref{prop:uhatToPC} is straightforward: from the 
constructed CPR in \cref{construct-cpr} we obtain the polynomial constraint
$\bigvee_{J=1}^{s_1} \varphi_{1,J}\wedge \langle \bt,D_{1,J}\bar{\bx}\rangle>0$
with a bounded number of alternations.

\paragraph{Step II: Replace real coefficients by rationals}
In our proof, the key step is to replace the real coefficients in the PC by rational coefficients so that the rational PC will define the same set of rational sequences, up to some given size. Let us make this precise. We denote by $\Q_{\le m}=\{a\in\Q \mid \normof{a}\le m\}$ the set of all rational numbers of size at most $m$. A polynomial constraint is \emph{rational} if all the polynomials occurring in it have rational coefficients.
\begin{proposition}\label{rational-pc}
	For every $m\in\N$ and every PC $\alpha$ with polynomials having $n$ variables, there exists a
	rational PC $\alpha'$ of polynomial size such that $\setof{\alpha}\cap
	\Q_{\le m}^{n}=\setof{\alpha'}\cap\Q_{\le m}^{n}$.
\end{proposition}

Proving \cref{rational-pc} requires the following technical lemma, for which we introduce some notation.
For two vectors $\bu,\bv\in\R^t$ and $m\in\N$, we write $\bu\sim_m\bv$ if for
every $\bw\in\Q^t$ and every $z\in\Q$ with $\normof{\bw},\normof{z}\le m$, we have
(i)~$\langle \bw,\bu\rangle\ge z$ if and only if $\langle \bw,\bv\rangle\ge z$
and (ii)~$\langle \bw,\bu\rangle> z$ if and only if $\langle \bw,\bv\rangle>z$.
In other words, we have $\bu\sim_m\bv$ if and only if $\bu$ and $\bv$ satisfy
exactly the same inequalities with rational coefficients of size at most $m$.
\begin{lemma}\label{equivalent-rational-vector}
	For every $\bc\in\R^t$ and $m\in\N$, there is a 
	$\bc'\in\Q^t$ with $\normof{\bc'}\le (mt)^{O(1)}$ and $\bc\sim_m\bc'$.
\end{lemma}

\begin{remark}
	For proving \cref{equivalent-rational-vector}, it is not sufficient to
	pick a rational $\bc'$ with $\|\bc'-\bc\|<\varepsilon$ for some small 
	enough $\varepsilon$. For example, note that if in some coordinate, $\bc$
	contains a rational number of size $\le m$, then in this coordinate,
	$\bc'$ and $\bc$ must agree exactly for $\bc\sim_m\bc'$ to hold.
\end{remark}
Before we prove \cref{equivalent-rational-vector}, let us see how to deduce \cref{rational-pc}: in a PC $\alpha$ we understand each polynomial $p(X_1,\dots,X_n)$ as a scalar product $\langle\bw,\bu\rangle$ where $\bw$ constains only variables and $\bu$ consists of all coefficients. Then \cref{equivalent-rational-vector} yields a vector $\bv\sim_{2m}\bu$ containing only rationals with the same behavior as $\bu$. From this we finally obtain polynomials having only rational coefficients, which also proves \cref{rational-pc}. A detailed proof of \cref{rational-pc} can be found in \cref{appendix-tc0}.

In the proof of \cref{equivalent-rational-vector}, we use the following fact about solution sizes to systems of inequalities. 
\begin{lemma}\label{rational-solution}
	Let $A\in\Q^{k\times n}$, $A'\in\Q^{\ell\times n}$, $\bz\in\Q^k$,
	$\bz'\in\Q^\ell$ with
	$\normof{A},\normof{A'},\normof{\bz},\normof{\bz'}\le m$.  If the
	inequalities $A\bx\gg\bz$ and $A'\bx\ge\bz'$ have a solution in $\R^n$,
	then they have a solution $\br\in\Q^n$ with $\normof{\br}\le (mn)^{O(1)}$.
\end{lemma}
We prove \cref{rational-solution} in the appendix. The proof idea is the following. By standard results about polyhedra, the set of vectors $\bx$ satisfying $A\bx\ge\bz$ and $A'\bx\ge\bz'$ can be written as the convex hull of some finite set $X=\{\bx_1,\ldots,\bx_s\}$, plus the cone generated by some finite set $\{\by_1,\ldots,\by_t\}$. Here, the vectors in $X\cup Y$ are all rational and of polynomial size.
By the Carath\'{e}odory Theorem, the real solution $\bs\in\R^n$ to $A\bs\gg\bz$ and $A'\bs\ge\bz'$ belongs to the convex hull of $n$ elements of $X$, plus a conic combination of $n$ elements of $Y$. We then argue that by taking the barycenter of those $n$ elements of $X$, plus the sum of the $n$ elements of $Y$ gives a rational vector $\br\in\Q^n$ with $A\br\gg\bz$ and $A'\br\ge\bz'$. The full proof of \cref{rational-solution} is in \cref{appendix-rational-solution}. To prove \cref{equivalent-rational-vector}, given $\bc\in\R^n$, we set up a system of (exponentially many) inequalities of polynomial size so that the solutions are exactly the vectors $\bd$ with $\bd\sim_m\bc$. The solution provided by \cref{rational-solution} is the desired $\bc'$ (see \cref{appendix-equivalent-rational-vector}).

\paragraph{Step III: Constructing $TC^0$ circuits}
It is now straightforward to translate a polynomial-sized CPR with rational coefficients and bounded alternations into a $\TCzero$ circuit:
\begin{proposition}\label{construct-tc0}
  Every language accepted by a UHAT with positional encoding is recognized by a family of circuits in $\TCzero$.
\end{proposition}

We now show that the $\TCzero$ upper bound is tight: There is a UHAT whose
language is $\TCzero$-hard under $\ACzero$ reductions. In particular, this
language is not in $\ACzero$, since $\ACzero$ is strictly included
in $\TCzero$.
\begin{proposition}
  There is a $\TCzero$-complete language that is accepted by a UHAT, even without positional encoding and masking,
  but is not recognized by any family of circuits in $\ACzero$.
\end{proposition}
\begin{proof}
  As shown by Buss~\cite[Corollary 3]{DBLP:journals/ipl/Buss92}, the problem of
	deciding whether $ab=c$ for given binary encoded integers $a,b,c\in\Z$
	is $\TCzero$-complete under $\ACzero$-reductions. Since $ab=c$ if and only if
	$ab>c-1$ and $-ab>-(c+1)$, the problem of deciding $ab>c$ is also
	$\TCzero$-complete. We exibit a UHAT such that
	the problem of deciding $ab>c$ can be $\ACzero$-reduced to membership in the
	language.

It suffices to define a UHAT $T$ that accepts a language $L \subseteq (\Q^2)^+$ 
such that for all $r,s \in \Q$ we have that $(r,s) \in L$ if and only if $r > s$.
Then we can reduce the test $ab>c$ to checking whether $(a,\frac{c}{b}) \in L$.
Note that formally, $\frac{c}{b}$ is represented as a string 
containing the binary encodings of $c$ and $b$ separated by a special symbol.
The UHAT $T$ is by definition initialized with the sequence $(1,r,s),(0,0,0) \in \Q^3$
since we only have to consider the accepted language restricted to sequences of length $1$.
It can directly check that $r-s > 0$ using the acceptance vector $\bm t := (0,1,-1)$.
\end{proof}

%
%

\section{UHAT and regular languages over infinite alphabets}
\label{sec:non-regular}
It was shown by \citet{ACY24} that UHATs with no positional encoding on binary 
input strings
accept only regular languages, even if masking is allowed.
We show that UHATs with masking over data sequences can recognize
``non-regular'' languages over infinite alphabet (\cref{thm:non-regular}). More
precisely, a standard 
notion of regularity over the alphabet $\ialphabet = \Q^d$ is that of 
\emph{symbolic 
automata} (see the CACM article \cite{DV21}), since it extends and 
shares all nice
closure and algorithmic properties of finite automata over finite alphabets,
while at the same time permitting arithmetics. Intuitively, a
transition rule in a symbolic automaton is of the form $p \to_\varphi q$,
where $\varphi$ represents the (potentially infinite) set $S \subseteq \Q^d$ of 
solutions to an
arithmetic constraint $\varphi$ (e.g. $2x = y$ represents $\{(n,2n) : n \in
\Q\}$). The meaning of such a transition rule is: move from
state $p$ to state $q$ by reading any $a \in S$.

To prove \cref{thm:non-regular}, we define the language
\[\mathsf{Double} := \{(r_1,\dots,r_n) \in \Q^n \mid n \ge 1 \text{ and } 2r_i < r_{i+1} \text{ for all } 1 \le i < n\}\]
of sequences of rational numbers where each number is more than double the 
preceding number.

\begin{lemma}
    UHAT with past masking and without positional encoding can recognize 
    $\mathsf{Double}$.
\end{lemma}
\begin{proof}
Given an input sequence $\bm u_1,\ldots,\bm u_{n+1}=(1,r_1),\ldots,(1,r_n),(0,0)
\in \Q^2$, we need to check that for all pairs $1\leq i<j<n+1$, we have $2\cdot
r_i<r_j$. To this end, a first standard encoder layer uses the differences
$2r_i-r_j$ as attention scores---except for $j=n+1$, where the attention score
will be $0$. This is achieved by setting the attention score for positions
$i,j$ to $2\bm u_i[2]\cdot \bm u_j[1]-\bm u_j[2]$. Indeed, this evaluates to
$2r_i-r_j$ for $i<j<n+1$, and to $0$ for $i<j=n+1$. In particular, for a
position $i\in[1,n]$, the attention score is maximized at $j=n+1$ if and only
if $2r_i<r_j$ for all $j\in[i+1,n]$.

The output vector $\bm v_i$ at position $i$ is then set to $\bm a_i[1]$, where
$\bm a_i$ is the attention vector at position $i$.  Thus, the output vector has
dimension $1$, and for each $i\in[1,n+1]$, we have $\bm v_i=0$ if and only if
$2\cdot r_i<r_j$ holds for all $j\in[i+1,n]$.
  
In a second standard encoder layer we now check whether all $\bm v_i$ have value $0$.
To this end, we choose for $i<j\leq n+1$ the attention score $\bm v_j$. Let
$\bm b_i$ denote the attention vector at position $i$. Then $\bm b_i=0$ iff
$\bm v_{i+1},\ldots,\bm v_n$ are all $0$. We then output $\bm w_i=1-(\bm
v_i+\bm b_i)$, which is positive if and only if $\bm v_i=\cdots=\bm v_{n+1}=0$.
Finally, with the acceptance vector $\bm t=1$ we accept if and only if $\bm
w_1>0$, which is equivalent to $\bm v_1=\cdots=\bm v_{n+1}=0$. As we saw above,
the latter holds if and only if $2r_i<r_j$ for all $i,j$ with $1\le i<j<n+1$.
\end{proof}
The proof of non-regularity of $\mathsf{Double}$ is easy (see \cref{app:non-regular}).
One could also easily show that $\mathsf{Double}$ 
cannot be recognized by other existing models 
in the 
literature of 
formal language theory over infinite alphabets, e.g., register
automata \cite{KF94,atom-book,DFSS19,ST11,TZ22}, variable/parametric automata
\cite{JLMR23,GKS10,FK18}, and data automata variants
\cite{two-variable,FL22}. For example, for register automata over $(\Q;<)$ (see
the book \cite{atom-book}), one could use the result therein that
data sequences accepted by such an automaton are closed under any 
order-preserving map of the elements in the sequence
(e.g., if $1,2,3$ is accepted, then so is $10,11,20$), which is not satisfied by
$\mathsf{Double}$.

\section{Logical languages accepted by UHAT}\label{sec:ltl}
In this section we show that an extension of linear temporal logic (LTL) with linear rational arithmetic (LRA)
and unary numerical predicates is expressible in UHAT over data sequences (\cref{thm:ltl}).
A formula of dimension $d > 0$ in \emph{locally testable LTL} (\ltltl) has the following syntax:
\[\varphi ::= \psi_k(\bm x_1,\dots,\bm x_{k+1}) \mid \Theta \mid \neg \varphi \mid \varphi \vee \varphi 
\mid X \varphi \mid \varphi U \varphi\]
Here, $\psi_k$ for $k \ge 0$ is an atom in LRA over the $d$-dimensional vectors of variables $\bm x_i$ of the form
$\langle\bm a,(\bm x_1,\dots,\bm x_{k+1})\rangle + b > 0$
where $\bm a \in \Q^{d(k+1)}$ and $b \in \Q$. 
Intuitively, $\psi_k$ allows one to check the values in the sequence with
$k$ ``lookaheads''.
Furthermore, $\Theta$ is a \emph{unary numerical predicate}, i.e. a family of functions 
$\theta_n \colon \{1,\dots,n\} \to \{0,1\}$
for all $n \ge 1$.
We define the satisfaction of an \ltltl\ formula $\varphi$ over a sequence 
$\bar{\bm v} = (\bm v_1,\dots,\bm v_n)$ of vectors in $\Q^d$ at position $i \in [1,n]$,
written $(\bar{\bm v},i) \models \varphi$, inductively as follows (omitting negation and disjunction):
\begin{itemize}
\item $(\bar{\bm v},i) \models \psi_k(\bm x_1,\dots,\bm x_{k+1})$ iff $i \le n-k$ and $\psi_k(\bm v_i,\dots,\bm v_{i+k})$
\item $(\bar{\bm v},i) \models \Theta$ iff $\theta_n(i) = 1$ 
\item $(\bar{\bm v},i) \models X \varphi$ iff $i < n$ and $(\bar{\bm v},i+1) \models \varphi$
\item $(\bar{\bm v},i) \models \varphi U \psi$ iff there is $j\in[i,n]$ with $(\bar{\bm v},j) \models \psi$ and
$(\bar{\bm v},k) \models \varphi$ for all $k\in[i,j-1]$
\end{itemize}
We write $L(\varphi) := \{\bar{\bm v} \in (\Q^d)^+ \mid (\bar{\bm v},1) \models \varphi\}$ for the language of $\varphi$.
\begin{example}
Consider sequences of the form 
$\bm x, A\bm x, A^2\bm x,\dots, A^n\bm x$
such that $\bm y A^n \bm x = 0$ and $n \ge 0$ is minimal with this property 
where $\bm y \in \Q^{1 \times d}$ and $A \in \Q^{d \times d}$ are fixed and $\bm x \in \Q^d$.
\cref{thm:ltl} implies that this language is   accepted by a UHAT since
it is defined by the \ltltl\ formula $G[(\neg\mathit{Last} \to (\bm y\bm x_1 \ne 0 \wedge A\bm x_1 = \bm x_2)) \wedge (\mathit{Last} \to \bm y\bm x_1 = 0)]$, where $\mathit{Last} := \neg X\top$.
\end{example}
\begin{example}
    Take the standard notion of 7-day Simple Moving Average (7-SMA); this can 
    be generalized to larger sliding windows of 50-days, or 100 days, which 
    are often used in finance. Using $\ltltl$, it is easy to show that the 
    following notion of ``uptrend'' can be captured using UHAT: sequences of 
    numbers such that the value at each time $t$ is greater than the 7-SMA 
    value at time $t$. The formula for this is:
    \[
        G(X^7\top \to \varphi(x_1,\ldots,x_7))
    \]
    where $\varphi(\bar x)$ is the formula $7x_7 > \sum_{i=1}^7 x_i$.
    Note here that $G\psi$ means (as usual for LTL) ``globally'' $\psi$, which
    can be written as $\neg (\top U \neg \psi)$. Similarly, $X^i$ means that $X$
    is repeated $i$ times.
\end{example}

We assume UHATs with positional encoding and a zero vector at the end of the input sequence (see \cref{sec:prelims}).
In the following we always assume that the components from the positional encoding are implicitly given
and are not changed by any UHAT.
So we write the sequence in \cref{eq:input} as $\bm v_1,\dots,\bm v_n,\bm 0$.
We use the following results from \citep{barcelo2024logical} that also hold for UHATs over data sequences.
\begin{lemma}\label{lem:iclr}
Let $d > 0$ and $\ell \in [1,d]$.
\begin{enumerate}[1)]
\item For every $b \in \{0,1\}$ there is a UHAT with positional encoding that on every sequence $\bm v_1,\dots,\bm v_n \in \Q^d$ 
with $v_i[\ell] \in \{0,1\}$ for all $i \in [1,n]$
outputs the sequence $\bm v_1,\dots,\bm v_{n-1},(\bm v_n[1,\ell-1],b,\bm v_n[\ell+1,d])$.\label{it:iclr-1}
\item There is a UHAT layer with positional encoding that on every sequence $\bm v_1,\dots,\bm v_n \in \Q^d$ and 
for every $i \in [1,n-1]$ picks attention vector $\bm a_i = \bm v_{i+1}$.\label{it:iclr-2}
\item There is a UHAT layer with positional encoding that on every sequence $\bm v_1,\dots,\bm v_n \in \Q^d$,
for every $\ell \in [1,d]$ with $\bm v_1[\ell],\dots,\bm v_{n-1}[\ell] \in \{0,1\}$ and $\bm v_n[\ell] = 0$, and for every $i \in [1,n]$
picks attention vector $\bm a_i = \bm v_j$ with minimal $j \in [i,n]$ such that $\bm v_j[\ell] = 0$.\label{it:iclr-3}
\end{enumerate}
\end{lemma}

Here, \ref{it:iclr-2} and \ref{it:iclr-3} directly follow from \citep{barcelo2024logical}.
For \ref{it:iclr-1} we remark that in \citep{barcelo2024logical} only the case $b = 0$ is shown.
On input $\bm v_1,\dots,\bm v_n$ as in \ref{it:iclr-1}, the UHAT uses positional encoding function $p(i,n) := (i,n)$ 
and a composition of affine transformations and ReLU to output at position $i \in [1,n]$ the vector
$(\bm v_i[1,\ell-1],b_i,\bm v_i[\ell+1,d])$ where 
$b_i := \min\{\bm v_i[\ell],n-i\}$ if $b = 0$ and 
$b_i := \max\{\bm v_i[\ell],i-n+1\}$ if $b = 1$.

Using \cref{lem:iclr}, we show that a UHAT can transform rational values $>0$ to 1
and values $\le 0$ to 0.
This will be used to evaluate inequalities by outputting 1 for true and 0 for false.

\begin{lemma}\label{lem:bools}
Let $d > 0$ and $\ell \in [1,d]$.
There is a UHAT with positional encoding that on every sequence $\bm v_1,\dots,\bm v_{n+1} \in \Q^d$
outputs $\bm v'_1,\dots,\bm v'_{n+1} \in \Q^d$ with
$\bm v'_i := (\bm v_i[1,\ell-1],b_i,\bm v_i[\ell+1,d])$ for all $i \in [1,n+1]$ where
$b_i := 1$ if $\bm v_i[\ell] > 0$ and $b_i := 0$ otherwise.
\end{lemma}
\begin{proof}
On input $\bm v_1,\dots,\bm v_{n+1} \in \Q^d$, the first layer
outputs at position $i \in [1,n+1]$ the vector $\bm w_i := (\bm v_i[1,\ell-1],r_i,\bm v_i[\ell+1,d])$ where
$r_i := \max\{\bm v_i[\ell],0\}$. Thus, 
$r_i = 0$ if $\bm v_i[\ell] \le 0$ and $r_i > 0$ otherwise.
The second layer turns the sequence $\bm w_1,\dots,\bm w_{n+1}$ into $(0,\bm w_1),\dots,(0,\bm w_{n+1})$.
We then apply \ref{it:iclr-1} of \cref{lem:iclr} to obtain the sequence $(0,\bm w_1),\dots,(0,\bm w_n),(1,\bm w_{n+1})$,
i.e.\ the last vector has first component 1, and all other vectors have first component 0.
Let $\bm u_1,\dots,\bm u_{n+1} \in \Q^{d+1}$ be the resulting sequence.
The final layer uses attention score
$\langle A \bm u_i, B \bm u_j \rangle$
for all $1 \le i,j \le n+1$ where the affine transformations $A,B \colon \Q^{d+1} \to \Q^{d+1}$ yield
$A \bm u_i = (\bm u_i[\ell],0,\dots,0)$ and
$B \bm u_j = (\bm u_j[1],0,\dots,0)$.
Let $\bm a_i$ be the attention vector of position $i \in [1,n+1]$.
Since $\bm u_i[\ell] \ge 0$, we have $\bm a_i[1] = 0$ if $\bm u_i[\ell] = 0$ 
and $\bm a_i[1] = 1$ if $\bm u_i[\ell] > 0$.
The layer outputs
$\bm v'_i := (\bm u_i[2,\ell],\bm a_i[1],\bm u_i[\ell+2,d+1])$
at position $i \in [1,n+1]$.
\end{proof}


    We now prove \cref{thm:ltl}.
We claim that for every \ltltl\ formula $\varphi$ of dimension $d$ and every $m \ge d$ 
there exists a UHAT $T_{\varphi,m}$ with positional encoding
that on every sequence $\bm w_1,\dots,\bm w_n,\bm 0 \in \Q^m$ outputs a sequence
$\bm w'_1,\dots,\bm w'_n,\bm 0 \in \Q^{m+1}$ such that for all $i \in [1,n]$ we have
$\bm w'_i[1,m] = \bm w_i$ and $\bm w'_i[m+1] = 1$ if $(\bar{\bm v},i) \models \varphi$ and $\bm w'_i[m+1] = 0$ otherwise,
where $\bar{\bm v} := (\bm w_1[1,d],\dots,\bm w_n[1,d])$.
Then the theorem follows since for every \ltltl\ formula $\varphi$ of dimension $d$ the UHAT $T_{\varphi,d}$ outputs
on every sequence $\bar{\bm v} = (\bm v_1,\dots,\bm v_n)$ of vectors in $\Q^d$ extended with the vector $\bm 0 \in \Q^d$ a sequence
$\bm v'_1,\dots,\bm v'_n, \bm 0 \in \Q^{d+1}$ such that $\bm v'_1[d+1] > 0$ if and only if $(\bar{\bm v},1) \models \varphi$.
Thus, $T_{\varphi,d}$ accepts $L(\varphi)$ by taking the acceptance vector
$\bm t := (0,\dots,0,1) \in \Q^{d+1}$. 

We prove the claim by induction on the structure of \ltltl\ formulas.
If the formula is a unary numerical predicate $\Theta$, then we can use the positional encoding
$p(i,n+1) := \theta_n(i)$ for all $i \in [1,n]$ and $p(n+1,n+1) := 0$ 
to output on every sequence $\bm w_1,\dots,\bm w_n,\bm 0 \in \Q^m$ the sequence
$(\bm w_1,p(1,n+1)),\dots,(\bm w_n,p(n,n+1)),(\bm 0,p(n+1,n+1)) \in \Q^{m+1}$.

If the formula is an atom $\psi_k(\bm x_1,\dots,\bm x_{k+1})$ of the form $\langle \bm a, (\bm x_1,\dots,\bm x_{k+1}) \rangle + b > 0$,
the UHAT $T_{\psi_k,m}$ adds in its first layer a component that is set to 1 to the top of every vector, 
outputting on every sequence $\bm w_1,\dots,\bm w_n,\bm 0 \in \Q^m$ the sequence $(1,\bm w_1),\dots,(1,\bm w_n),(1,\bm 0)$.
Then we apply \ref{it:iclr-1} of \cref{lem:iclr} to turn this sequence into $(1,\bm w_1),\dots,(1,\bm w_n),\bm 0$.
Next, $T_{\psi_k,m}$ uses $k$ layers to allow each position to gather the first $d+1$ components of its $k$ right neighbors.
More precisely, the $\ell$-th layer, for $\ell \in [1,k]$, on sequence $\bm u_1,\dots,\bm u_{n+1}$ 
uses \ref{it:iclr-2} of \cref{lem:iclr} to get for every position $i \in [1,n]$ the attention vector $\bm a_i = \bm u_{i+1}$
and the attention vector $\bm a_{n+1}$ is arbitrary.
Note that if $\ell = 1$, then $\bm a_n = \bm 0$.
Then it applies an affine transformation to output at position $i \in [1,n]$ the vector $(\bm a_i[1,d+1],\bm u_i)$
and using \ref{it:iclr-1} at position $n+1$ the vector $(0,\bm a_{n+1}[2,d+1],\bm u_{n+1})$.
Let $\bm u_1,\dots,\bm u_{n+1} \in \Q^{k(d+1)+m+1}$ be the output of the $k$-th of those layers.
We add another layer that using a composition of ReLU and affine functions outputs at every position $i \in [1,n+1]$ the vector 
$\bm u'_i := (\bm u_i[k(d+1)+2,k(d+1)+m+1],r_i) \in \Q^{m+1}$ where
$r_i := \min\{\bm u_i[1], \langle \bm a, \hat{\bm u}_i \rangle + b\}$
and\ifSubmission{\pagebreak}
\[\hat{\bm u}_i := (\bm u_i[2,d+1],\bm u_i[d+3,2(d+1)],\dots,\bm u_i[kd+k+2,(k+1)(d+1)]),\]
which contains the first $d$ components of the initial input vector and its $k$ right neighbors.
That is, for all $i \in [1,n]$ we have that
$r_i = 0$ if $\langle \bm a, \hat{\bm u}_i \rangle + b \le 0$ or
$i+k > n$ since $\bm u_i[1]$ can only be 0 if it was gathered from the vector at position $n+1$ using attention.
Furthermore, $r_i > 0$ if $\langle \bm a, \hat{\bm u}_i \rangle + b > 0$ and $i+k \le n$.
Note that $\bm u'_i[1,m]$ is equal to the input vector $\bm w_i$ from the beginning if $i \in [1,n]$ and $\bm 0$ if $i = n+1$.
Finally, we apply \cref{lem:bools} followed by \ref{it:iclr-1} to output at position $i \in [1,n]$ the vector 
$\bm w'_i := (\bm u'_i[1,m],r'_i)$ where $r'_i := 1$ if $r_i > 0$ and $r'_i := 0$ otherwise
and at position $n+1$ the vector $\bm w'_{n+1} := (\bm u'_{n+1}[1,m],0) = \bm 0$.
Thus, for all $i \in [1,n]$ we have that $\bm w'_i[m+1] = 1$ if $(\bar{\bm v},i) \models \psi_k$ and $\bm w'_i[m+1] = 0$ otherwise,
where $\bar{\bm v} := (\bm w_1[1,d],\dots,\bm w_n[1,d])$.

Let us now continue with the inductive step where we assume that $\varphi$ and $\psi$ are \ltltl\ formulas of dimension $d$
such that for all $m \ge d$ we already showed existence of the UHATs $T_{\varphi,m}$ and $T_{\psi,m+1}$.
For the cases $\neg\varphi$, $\varphi \vee \psi$, and $X\varphi$ we refer to \cref{app:ltl}.
For $\varphi U \psi$ define the UHAT $T_{\varphi U \psi,m}$ that first applies $T_{\psi,m+1} \circ T_{\varphi,m}$
outputting a sequence $\bm u_1,\dots,\bm u_n, \bm 0 \in \Q^{m+2}$.
Observe that $(\bar{\bm v},i) \models \varphi U \psi$ for $i \in [1,n]$ and $\bar{\bm v} := (\bm u_1[1,d],\dots,\bm u_n[1,d])$ 
if and only if 
for the minimal $j \in [i,n]$ with $(\bar{\bm v},j) \models \neg\varphi \vee \psi$ we have $(\bar{\bm v},j) \models \psi$
and such a $j$ exists.
Equivalently, for the minimal $j \in [i,n+1]$ with $\min\{\bm u_j[m+1],1-\bm u_j[m+2]\} = 0$ it holds that $\bm u_j[m+2] = 1$.
To check this, we first add a layer that outputs at position $i \in [1,n+1]$ the vector
$\bm u'_i := (\bm u_i,\min\{\bm u_i[m+1],1-\bm u_i[m+2]\}) \in \Q^{m+3}$.
Finally, we add a layer that uses \ref{it:iclr-3} of \cref{lem:iclr} to get attention vector
$\bm a_i = \bm u'_j$ with $j \in [i,n+1]$ minimal such that $\bm u'_j[m+3] = 0$.
The layer then outputs at position $i \in [1,n+1]$ the vector
$\bm w'_i := (\bm u'_i[1,m],\bm a_i[m+2]) \in \Q^{m+1}$.

\section{Concluding remarks}
\label{sec:conc}

We initiated the study of the expressive power of
transformers, when the input is a sequence of (tuples of) numbers, which is the
setting for applications like time series analysis/forecasting. Our results
indicate an increased expressiveness of transformers on such input data, 
in comparison to the previous formal language theoretic 
setting (see survey \citep{transformers-survey}), i.e., when a 
token
embedding function (with a bounded number of tokens) is first applied before
feeding the input to a transformer. More precisely, this represents for Unique
Hard Attention Transformers (UHAT) a jump from
the complexity class $\ACzero$ to $\TCzero$ (since $\ACzero \subsetneq
\TCzero$), and the jump from regular to non-regular languages (when position 
encoding is not allowed). On the positive side, we successfully developed
an expressive class of logical languages recognized by UHAT in terms of a logic
called locally testable LTL, which extends previously identified logic for UHAT
for strings over finite alphabets \citep{ACY24,barcelo2024logical}.

\textbf{Limitations.} While we follow the standard 
formalization of transformer encoders in Formal Languages and 
Neural Networks (e.g. 
\citep{transformers-survey,barcelo2024logical,Hao22,Hahn20}), limitations
of the models are known (see \cite{transformers-survey}
for a thorough discussion). 
For example, used real numbers could be of 
\emph{unbounded} precision, which allow one to precisely represent
values of $\sin$ and $\cos$ functions (actually used in practice for
positional encoding). In addition,
the positional encoding used by the model could be \emph{uncomputable}.
Three answers
can be given. First, an \emph{upper bound complexity} on the model 
with unbounded precision and arbitrary positional encodings
(e.g. in $\TCzero$) still applies in the case of
bounded precision. Second, limiting the power of UHAT (e.g. allow only
rational numbers, and assuming 
efficient (i.e.\ uniform $\TCzero$) computability of the 
positional encoding $p: \N \times \N \to \Q^d$), our proof in fact yields 
\emph{uniformity} of our $\TCzero$ upper bound. Third, our lower bound for
non-regularity of UHAT (cf. \cref{thm:non-regular}) holds even with only
rational numbers and no positional encodings. Finally, to alleviate these
issues, 
we have always made an explicit distinction between UHAT with and without positional encodings.

\textbf{Future directions.} Our paper opens up a plethora of research 
avenues on the expressive power 
of transformers on data sequences. In particular, one could consider other 
transformer encoder
models that have been considered in the formal language theoretic setting to
transformers (see the survey \citep{transformers-survey}). For example, instead 
of unique hard attention mechanism, we 
could consider the expressive power of transformers on 
data sequences using \emph{average hard attention} mechanism. Similar question
could be asked if we use a softmax function instead of a hard attention,
which begs the question of which numerical functions could be computed in
different circuit complexity classes like $\TCzero$.
Another important question concerns a logical characterization for UHAT over
sequences of numbers. 
This is actually still an open question even for the case of finite alphabets. 
Barcelo et al.~\cite{barcelo2024logical} 
showed that first-order logic (equivalently, LTL) with monadic numerical 
predicates (called LTL(Mon)) is subsumed in UHAT with arbitrary position 
encodings, i.e., the transformer model that we are generalizing in this paper 
to sequences of numbers.
There are UHAT languages (e.g.\ the set of palindromes) that are not captured by this. As remarked in
\cite{barcelo2024logical}, LTL(Mon) can be extended with arbitrary linear 
orders on the positions (parameterized by lengths), which then can define 
the palindromes. [An analogous extension for $\ltltl$ can define palindromes
over an infinite alphabet.] However, it is possible to show 
that the resulting logic is still not expressive enough to capture the full generality of UHAT.
That said, although our logic does not capture the full UHAT, it can still be 
used to see at a glance what languages can be recognized by UHAT (e.g.\ Simple
Moving Averages).
There could perhaps be a hope of obtaining a precise logical characterization 
if we restrict the model of UHAT. The recent paper \cite{ACY24} showed that 
LTL(Mon) captures precisely the languages of masked UHAT with position 
encodings with ``finite image''. It is interesting to study similar 
restrictions for UHAT in the case of sequences of numbers.

\ifFull{%
\paragraph{Acknowledgments}
\raisebox{-9pt}[0pt][0pt]{\includegraphics[height=.8cm]{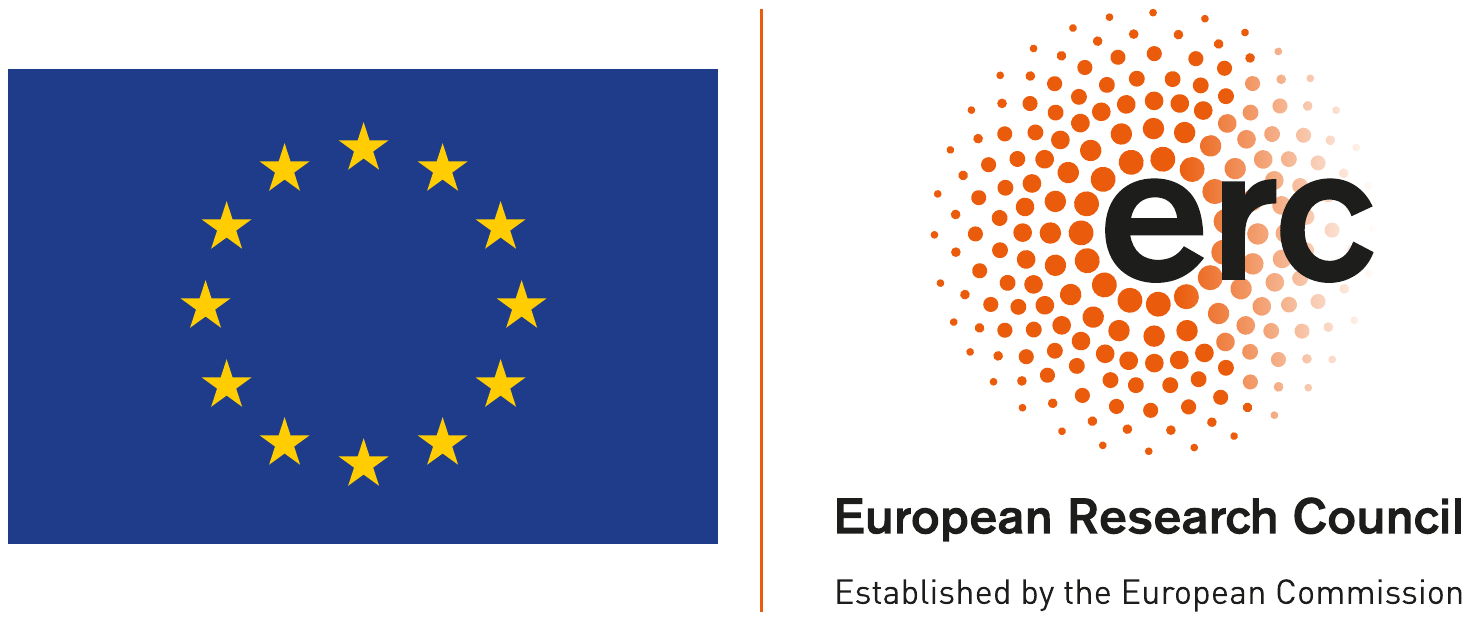}}
Funded by the European Union (ERC, LASD, 101089343 and FINABIS, 101077902). 
Views and opinions expressed are however those of the authors only and do not necessarily reflect those of the European Union 
or the European Research Council Executive Agency. 
Neither the European Union nor the granting authority can be held responsible for them.

}
\ifConference{%
\paragraph{Acknowledgments}
\raisebox{-9pt}[0pt][0pt]{\includegraphics[height=.8cm]{eu-erc.pdf}}
Funded by the European Union (ERC, LASD, 101089343 and FINABIS, 101077902). 
Views and opinions expressed are however those of the authors only and do not necessarily reflect those of the European Union 
or the European Research Council Executive Agency. 
Neither the European Union nor the granting authority can be held responsible for them.

}

\label{beforebibliography}
\newoutputstream{pages}
\openoutputfile{main.pages.ctr}{pages}
\addtostream{pages}{\getpagerefnumber{beforebibliography}}
\closeoutputstream{pages}
\bibliographystyle{apalike}
\bibliography{refs}

\ifFull{%
\newpage

\appendix

\section{Example UHAT with real parameters}\label{appendix-example-real-parameters}
We present here an example UHAT in which two real numbers $\alpha$ and $\beta$ occur (each occurs once in a matrix associated to a particular layer) such that the simple sequence 
\begin{equation} (1,0,\be_1)(1,0,\be_2)(1,0,\be_3), \label{example-sequence}\end{equation}
where $\be_i\in\R^3$ is the $i$-th unit vector,
 is accepted if and only if $\alpha\beta=2$ and $\alpha=\beta$. This shows that even if we restrict the input sequence to a particular number $B$ of bits (e.g. the number of bits to represent the input sequence \cref{example-sequence}), it is not possible to replace the real constants $\alpha$ and $\beta$ by rational numbers without changing the accepted sequences of up to $B$ bits: The sequence \cref{example-sequence} is only accepted if $\alpha=\beta=\sqrt{2}$. And if we change $\alpha$ or $\beta$ in any way, the sequence \cref{example-sequence} will not be accepted anymore.

\begin{enumerate}
\item Input layer: 
\[ (1,0,\be_1)(1,0,\be_2)(1,0,\be_3). \]
\item Using a standard encoding layer, we multiply the first component in each position by $\alpha$.
Result: 
\[ (\alpha,0,\be_1)(\alpha,0,\be_2)(\alpha,0,\be_3). \]
\item Using attention, we can apply distinct affine transformations to the first three positions. We choose the following affine transformations. The first position is unchanged. The second position is mapped to $(1,0,\be_2)$, and the third position is mapped to $(0,\alpha,\be_3)$, using a matrix that flips the first and second component. Result: 
\[ (\alpha,0,\be_1)(1,0,\be_2)(0,\alpha,\be_3). \]
\item Using a standard encoding layer, we multiply the first component in each position with $\beta$.
Result: 
\[ (\alpha\beta,0,\be_1)(\beta,0,\be_2)(0,\alpha,\be_3). \]
\item Finally, by using our result on $\ltltl$, we can build further layers so that we accept if and only if (i)~the first component of the first position equals $2$ and (ii)~the first component of the second position equals the second component of the third position. Thus, we accept our original input if and only if $\alpha\beta=2$ and $\alpha=\beta$.
\end{enumerate}

\section{Omitted definitions and proofs in Section~\ref{sec:tc0}} \label{appendix-tc0}

\subsection{Descriptional size}
\newcommand{\size}[1]{\operatorname{size}(#1)}

Let $x\in\R$ be some real number. The \emph{(descriptional) size} of $x$ is $\size{x}=1+\lceil\log_2(|p|+1)\rceil+\lceil\log_2(|q|+1)\rceil$ if $x=\tfrac{p}{q}$ is a rational number (where $p$ and $q$ are relatively prime integers) and $\size{x}=1$ if $x$ is an irrational number. Note that in the latter case we use the number $1$ as some placeholder, since we do not have to represent irrational numbers in any algorithm. However, for analysis of the sizes in our constructions we still need some value.

Let $\bv\in\R^d$ be a vector. The \emph{size} of $\bv$ is $\size{\bv}=n+\sum_{i=1}^d\size{v_i}$ where $\bv=(v_1,\dots,v_d)^T$.

Let $M\in\R^{m\times n}$ be a matrix. The \emph{size} of $M$ is $\size{M}=mn+\sum_{1\leq i\leq m,1\leq j\leq n}\size{a_{ij}}$ where $M=(a_{ij})_{1\leq i\leq m,1\leq j\leq n}$.

Let $A\colon\R^{d}\to\R^e$ be an affine transformation, i.e., we have $A(\bx)=B\bx+\bc$ with $B\in\R^{d\times e}$ and $\bc\in\R^e$. Then the \emph{size} of $A$ is $\size{A}=\size{B}+\size{c}+1$.

Now, let $p\in\R[X_1,\dots,X_n]$ be a polynomial. Then we have $p(X_1,\dots,X_n)=\sum_{0\leq r_1,\dots,r_n\leq k}c_{r_1,\dots,r_n}X_1^{r_1}\cdots X_n^{r_n}$ for some numbers $k\in\N$ and $c_{r_1,\dots,r_n}\in\R$. The \emph{size} of $p$ is
\[\size{p}=\sum_{0\leq r_1,\dots,r_n\leq k}\size{c_{r_1,\dots,r_n}}+\size{r_1}+\dots+\size{r_n}+n\,.\]

Let $\alpha$ be a polynomial constraint. We define the \emph{size} of $\alpha$ inductively on the structure of the formula as follows:
\begin{itemize}
  \item if $\alpha=(p(X_1,\dots,X_n)\sim0)$ with $\mathord{\sim}\in\{>,\geq\}$ is an atom. Then $\size{\alpha}=\size{p}+1$.
  \item if $\alpha=\bigwedge_{1\leq i\leq k}\beta_i$ or $\alpha=\bigvee_{1\leq i\leq k}\beta_i$ is a formula with PCs $\beta_1,\dots,\beta_k$, then $\size{\alpha}=k+\sum_{1\leq i\leq k}\size(\beta_i)$.
\end{itemize}

Let $R=(\phi_1,D_1),\dots,(\phi_k,D_k)$ be a CPR. Then the \emph{size} of this CPR is $\size{R}=k+\sum_{i=1}^{k}\size{\phi_i}+\size{D_k}$.

\subsection{ReLU case in Lemma~\ref{construct-cpr}}
Let us now consider a ReLU layer. Assume that we compute the ReLU-value for the $j$-th component, i.e., we compute $\max\{0,x_{i,j}\}$ for each $i\in\{1,\dots,n+1\}$ where $x_{i,j}$ is the $j$-th component of $\bx_i$. From each conditional assignment $\varphi_{i,k}\to D_{i,k}$ with $i\in\{1,\dots,n+1\}$ and $1\leq k\leq s_i$ we construct two new conditional assignments:
\begin{enumerate}
  \item $\langle\be_{(i-1)\cdot n+j},D_{i,k}\bar{\bx}\rangle\geq0\to D_{i,k}$ where $\be_h$ is the $h$-th unit vector.
  \item $\langle-\be_{(i-1)\cdot n+j},D_{i,k}\bar{\bx}\rangle>0\to MD_{i,k}$ where $M=(m_{gh})_{1\leq g,h\leq d'\cdot (n+1)}$ is the matrix with $m_{gh}=1$ if $g=h\neq(i-1)\cdot n+j$ and $m_{gh}=0$ otherwise (i.e., $M$ is the unit matrix except for the $(i-1)\cdot j$-th entry).
\end{enumerate}
Now, if the $j$-th component of $\bx_i$ is non-negative, only the first conditional assignment is satisfied and the value of this component is left untouched. Otherwise, the $j$-th component is negative. But then the second conditional assignment is satisfied and the value of this component is set to $0$ (while the others stay unchanged). So, we obtain again some polynomial sized CPR with the same number of alternations as before.

\subsection{Omitted proofs}
\begin{proof}[Proof of \cref{prop:uhatToPC}]
  Let $\bt\in\R^{e}$ be the acceptance criterion of the UHAT and $f\colon\R^{d\cdot (n+1)}\to\R^{e\cdot (n+1)}$ be the computed function for inputs of length $n$.
  By \cref{construct-cpr}, we can construct in polynomial time an $O(\ell)$-alternation-bounded CPR computing $f$. So, let $\varphi_{i,k}\to D_{i,k}$ be the conditional assignments in this CPR (for $1\leq i\leq n+1$ and $1\leq k\leq s_i$). Then we obtain a PC from this CPR as follows:
  \[ \bigvee_{J=1}^{s_1} \varphi_{1,J}\wedge \langle \bt,D_{1,J}\bar{\bx}\rangle>0\,. \]
  Note that this PC has still polynomial size, accepts an input sequence $(\bx_1,\dots,\bx_n,\bm0)$ if, and only if, the UHAT accepts $(\bx_1,\dots,\bx_n)$, and --- if the CPR is $a$-alternation-bounded --- then it has at most $a+2$ alternations of disjunctions and conjunctions.
\end{proof}

\begin{proof}[Proof of \cref{rational-pc}]
  Consider a constraint $p(X_1,\dots,X_n)>0$ (or $\geq0$ resp.) in $\alpha$. Then $p\in\R[X_1,\dots,X_n]$ is a polynomial of degree at most $2$, i.e., there are real numbers $c_{i,j,r,s}\in\R$ such that
  \[p(X_1,\dots,X_n)=\sum_{0\leq r+s\leq 2}\sum_{1\leq i\leq j\leq n}c_{i,j,r,s}X_i^rX_j^s\,.\]
  Now, construct two vectors $\bu$ and $\bw$ with a component for each tuple $(i,j,r,s)$: $u_{i,j,r,s}=c_{i,j,r,s}$ and $w={i,j,r,s}=X_i^rX_j^s$. Then it is clear that $p(X_1,\dots,X_n)=\langle\bu,\bw(X_1,\dots,X_n)\rangle$ holds.
  
  Application of \cref{equivalent-rational-vector} yields a vector $\bv\in\Q^t$ with $\normof{\bv}\leq(2mt)^{O(1)}$ and $\bu\sim_{2m}\bv$. Note that we need to consider rational numbers up to size $2m$ due to the fact that substitution of the variables in $\bw$ by rational numbers $\bx\in\Q^n_{\leq m}$ yields a rational vector $\bw(\bx)\in Q^t_{\leq2m}$. Let $p'(X_1,\dots,X_n)$ be the polynomial obtained from $p$ be replacing the coefficients $c_{i,j,r,s}\in\R$ by $v_{i,j,r,s}\in\Q$. Then for each $\bx\in\Q^n_{\leq m}$ we have $p(\bx)=\langle\bu,\bw(\bx)\rangle>0$ (resp. $\geq0$) if, and only if, $p'(\bx)=\langle\bv,\bw(\bx)\rangle>0$ (resp. $\geq0$).
  
  Replacing each real polynomial $p$ in $\alpha$ by the constructed rational polynomial $p'$ results in a rational PC $\alpha'$ with $\setof{\alpha}\cap\Q^n_{\leq m}=\setof{\alpha}\cap\Q^n_{\leq m}$.
\end{proof}

\subsection{Proof of Lemma~\ref{rational-solution}}\label{appendix-rational-solution}
\newcommand{\convhull}[1]{\mathsf{conv.hull}~#1}
\newcommand{\cone}[1]{\mathsf{cone}~#1}
The rest of this subsection is devoted to proving \cref{equivalent-rational-vector}, for which we rely on results from convex geometry, which
requires some terminology. For a set $S\subseteq\R^n$, we define the \emph{convex hull of $S$} as
\begin{multline*} \convhull{S}=\{\lambda_1\bs_1+\cdots+\lambda_m\bs_m \mid m>0,~\bs_1,\ldots,\bs_m\in S,\\\lambda_1,\ldots,\lambda_m\in[0,1],~\lambda_1+\cdots+\lambda_m=1\} \end{multline*}
and the \emph{cone generated by $S$} as
\begin{multline*} \cone{S}=\{\lambda_1\bs_1+\cdots+\lambda_m\bs_m \mid m>0,~\bs_1,\ldots,\bs_m\in S,~\lambda_1,\ldots,\lambda_m\ge 0\}. \end{multline*}
We will also rely on Carath\'{e}odory's theorem, which says that a points in cones and convex sets can be obtained from at most $n$ points. See, for example, \cite[Theorems 7.1i and 7.1j]{Schrijver1986}.
\begin{theorem}[Carath\'{e}odory's Theorem]
	Let $S\subseteq\R^n$. For every $\bx\in\convhull{S}$, there are
	$\bx_1,\ldots,\bx_n\in S$ with $\bx\in\convhull{\{\bx_1,\ldots,\bx_n\}}$. Moreover, for
	every $\by\in\cone{S}$, there are $\by_1,\ldots,\by_n\in S$ with
	$\by\in\cone{\{\by_1,\ldots,\by_n\}}$.
\end{theorem}
A \emph{polyhedron} is a set of the form $\{\bx\in\R^n \mid A\bx\ge\bb\}$,
where $A\in\R^{m\times n}$ and $\bb\in\R^m$ for some $m\in\N$. If the matrix
$A$ and the vector $\bb$ are rational, then the polyhedron is called a
\emph{rational polyhedron}. It is a standard result about polyhedra that if $A$ and $\bb$ are rational of size at most $m$, then the polyhedron $P=\{\bx\in\R^n\mid A\bx\ge\bb\}$ can be written as
\[ P=\convhull\{\bx_1,\ldots,\bx_s\}+\cone\{\by_1,\ldots,\by_t\}, \]
with rational vectors $\bx_1,\ldots,\bx_s,\by_1,\ldots,\by_t$, where each
vector has size polynomial in $mn$. See, for example, \cite[Theorem
10.2]{Schrijver1986}. 
\begin{proof}
	We define the matrix $B\in \Q^{(k+\ell)\times n}$ and the vector $\bb\in\Q^{k+\ell}$ as
	\begin{align*}
		B &= \begin{bmatrix} A \\ A'\end{bmatrix} & \bb=\begin{bmatrix} \bz \\ \bz'\end{bmatrix}.
	\end{align*}
	and consider the polyhedron $P=\{\bx\in\R^n \mid B\bx\ge \bb\}$. As mentioned above, we can write
	\[ P=\convhull\{\bx_1,\ldots,\bx_s\}+\cone\{\by_1,\ldots,\by_t\}, \]
	where $\bx_1,\ldots,\bx_s,\by_1,\ldots,\by_t$ are rational vectors of
	size polynomial in $mn$. By our assumption, there exists an $\bs\in\R^n$ with $A\bs\gg\bz$ and $A'\bs\ge\bz'$. By Carath\'{e}odory's Theorem, wlog, $\bs$ belongs to the smaller polyhedron
	\[ Q=\convhull\{\bx_1,\ldots,\bx_n\}+\cone\{\by_1,\ldots,\by_n\}. \]
	Now note that we have $A\bu\ge\bz$ and $A'\bu\ge\bz'$ for every
	$\bu\in U:=\{\bx_1,\ldots,\bx_n,\by_1,\ldots,\by_n\}$.
	We claim that the vector 
	\[ \br=\tfrac{1}{n}(\bx_1+\cdots+\bx_n)+\by_1+\cdots+\by_n \]
	satisfies $A\br\gg\bz$ and $A'\br\ge\bz'$. Indeed, it clearly belongs
	to $Q\subseteq P$ and thus satisfies $A'\br\ge\bz'$. Moreover, for every
	row $\ba^\top \bx>z$ of $A\bx\gg\bz$, there must be a $\bu\in U$ with
	$\ba^\top\bu>z$---otherwise, we would would have $\ba^\top\bu=z$ for
	every $\bu\in U$ and thus $\ba^\top\bs=z$. In particular, we have
	$A\br\gg\bz$.  Finally, the vector $\br$ has size at most
	$\normof{\bx_1}+\cdots+\normof{\bx_n}+\normof{\by_1}+\cdots+\normof{\by_n}$,
	which is polynomial in $mn$, since each $\bx_i$ and each $\by_i$ has size polynomial in $mn$.
\end{proof}

\subsection{Proof of Lemma~\ref{equivalent-rational-vector}}\label{appendix-equivalent-rational-vector}
\begin{proof}[Proof of \cref{equivalent-rational-vector}]
	Collect the set of all inequalities $\langle \bw,\bu\rangle>z$ or
	$\langle \bw,\bu\rangle\ge z$ with $\bw\in\Q^t_{\le m}$ and
	$z\in\Q_{\le m}$ that are satisfied for $\bu$.  This results in two
	large matrices $A\in\Q^{k\times t}_{\le m}$ and $A'\in\Q^{\ell\times
	t}_{\le m}$ and vectors $\bz\in\Q_{\le m}^k$ and $\bz'\in\Q^{\ell}_{\le
	m}$ such that we have $\bu\sim_m\bv$ if and only if $A\bv\gg\bz$ and
	$A'\bv\ge\bz'$. Thus, we can construct $\bc'$ using
	\cref{rational-solution}.  Observe that the bound from
	\cref{rational-solution} does not depend on the (exponentially large)
	$k$ and $\ell$.
\end{proof}

\begin{proof}[Proof of \cref{construct-tc0}]
  Let $T$ be some UHAT with positional encoding and $n,m\in\N$ be two natural numbers. In the following, we consider input sequences of $T$ having the length $n$ and size $m$. By \cref{prop:uhatToPC,rational-pc} there is polynomial sized, $O(\ell)$-alternation-bounded, and rational PC $\alpha$ (where $\ell$ is the number of layers in $T$) such that the set of sequences of size $m$ accepted by the UHAT $T$ equals $\setof{\alpha}\cap\Q^{n\cdot d}_{\leq m}$.
  
	We finally show that the PC $\alpha$ can be realized as a circuit of constant depth and polynomial size. So, consider a constraint of the form $p(\bar{\bx}) \sim 0$ where $\mathord{\sim} \in\{\geq,>\}$, $p$ is a polynomial of degree at most $2$ and $\bar{\bx}$ represents the input sequence. Since addition and multiplication of rational numbers are realizable in $\TCzero$~\cite{chandra1984constant}, it is clear that the computation of the value $p(\bar{\bx})$ is also realizable. Additionally, checking whether this value is $\geq0$ (or $>0$, resp.) is a simple check of the bit representing the signum (and checking that the numerator has at least one non-zero bit).
  
  Finally, we have to connect all the atoms of the form $p(\bar{\bx}) \sim 0$
to the Boolean formula $\alpha$. Since $\alpha$ alternates only a bounded
number of times between disjunctions and conjunctions, we can realize the
complete formula $\alpha$ in a circuit of constant depth and with polynomial
size. Since $\alpha$ is equivalent (up to the input size $m$), the UHAT $T$ is
realizable in $\TCzero$. Note that here, we need to perform iterated addition
of rational numbers, which requires iterated multiplication of integers
represented in binary. The latter is well-known to be possible in
$\TCzero$~\cite{DBLP:conf/coco/Reif87,DBLP:journals/siamcomp/ReifT92} (even
uniformly~\cite{DBLP:journals/jcss/HesseAB02}, but this is not needed in our
setting).
\end{proof}

\section{Proof of non-regularity in Section \ref{sec:non-regular}}
\label{app:non-regular}

Recall that 
\[\mathsf{Double} := \{(r_1,\dots,r_n) \in \Q^n \mid n \ge 1 \text{ and } 2r_i <
r_{i+1} \text{ for all } 1 \le i < n\}.
\]

We next define the notion of symbolic automata \cite{DV21,popl12,DV17}.
A \emph{symbolic automaton} is a tuple $(Q,\delta,q_0,F)$, where $Q$ is a finite
set of states, $q_0 \in Q$ is an initial state, $F \subseteq Q$ is a set of
final states, and $\delta$ is a set of transition rules of the form $(p,S,q)$,
where $S \subseteq \Q$. For $a \in \Q$, we write $p \to_a q$ (read ``there
is a transition from $p$ to $q$ reading $a$) if there is a
transition rule $(p,S,q)$ such that $a \in S$. Slightly abusing notation, for a
set $S \subseteq \Q$, we also write $p \to_S q$ to mean that $(p,S,q)$ is a
transition rule in $\delta$.  The notion of a run, and an
accepting run can then be defined in exactly the same way as for finite
automata (e.g. see \cite{Sipser-book}); namely, it is a sequence of transitions
$q_0 \to_{a_1} \cdots \to_{a_n} q_n$, where $q_n \in F$, reading the sequence 
$w = a_1\cdots a_n$ over $\Q$. 

To prove that there is no symbolic automaton recognizing $\mathsf{Double}$, let
us assume to the contrary that such an automaton $A$ exists, say, with $n$
states. Consider a sufficiently long $w = a_1\cdots a_m \in \mathsf{Double}$ 
(i.e. of length at least $n$), and an accepting run of $A$:
\[
    q_0 \to_{S_1} \cdots \to_{S_m} q_m
\]
where each $a_i \in S_i$ and $q_m \in F$. 
Since $m+1 > n$, by pigeonhole principle, it must be the case that $q_r = q_s$ for some $r < s$.
Thus, also the sequence
\[a_1\cdots a_r (a_{r+1} \cdots a_s)^2 a_{s+1} \cdots a_m\]
is accepted by $A$ and therefore contained in $\mathsf{Double}$.
This is a contradiction since $\mathsf{Double}$ imposes $a_s < a_s$.

\section{Omitted cases in proof of Theorem \ref{thm:ltl}}\label{app:ltl}
For $\neg\varphi$ the UHAT $T_{\neg\varphi,m}$ first applies $T_{\varphi,m}$ and on the obtained sequence
$\bm u_1,\dots,\bm u_n,\bm 0 \in \Q^{m+1}$ uses a further layer followed by \ref{it:iclr-1} to output
$\bm w'_i := (\bm u_i[1,m],1-\bm u_i[m+1])$ at position $i \in [1,n]$ and
$\bm 0 \in \Q^{m+1}$ at position $n+1$.

For $\varphi \vee \psi$ we define the UHAT $T_{\varphi \vee \psi,m}$ that first applies $T_{\psi,m+1} \circ T_{\varphi,m}$
followed by a layer that on sequence $\bm u_1,\dots,\bm u_{n+1} \in \Q^{m+2}$ outputs $\bm w'_1,\dots,\bm w'_{n+1} \in \Q^{m+1}$ with
\[\bm w'_i := (\bm u_i[1,m],\max\{\bm u_i[m+1],\bm u_i[m+2]\})\] 
for all $i \in [1,n+1]$.
Note that if $\bm u_{n+1} = \bm 0$, then also $\bm w'_{n+1} = \bm 0$.

For $X \varphi$ the UHAT $T_{X\varphi,m}$ first applies $T_{\varphi,m}$ to output a sequence
$\bm u_1,\dots,\bm u_{n+1}\in \Q^{m+1}$ with $\bm u_{n+1} = \bm 0$.
With an additional layer that uses \ref{it:iclr-2} to get attention vector
$\bm a_i = \bm u_{i+1}$ for all $i \in [1,n]$ it then outputs at position $i \in [1,n]$ the vector
$\bm w'_i := (\bm u_i[1,m],\bm a_i[m+1])$
and at position $n+1$ the vector $\bm w'_{n+1} := (\bm u_{n+1}[1,m],0)$ after applying \ref{it:iclr-1}.

}

\ifSubmission{%

\newpage
\section*{NeurIPS Paper Checklist}

\begin{enumerate}
  
  \item {\bf Claims}
  \item[] Question: Do the main claims made in the abstract and introduction accurately reflect the paper's contributions and scope?
  \item[] Answer: \answerYes{} 
  \item[] Justification: This paper contains theoretical results on the 
  expressive power of transformers on data sequences. As standard in
  theoretical results, we try to mathematically formulate the results as 
  as precisely as they can be.
  \item[] Guidelines:
  \begin{itemize}
    \item The answer NA means that the abstract and introduction do not include the claims made in the paper.
    \item The abstract and/or introduction should clearly state the claims made, including the contributions made in the paper and important assumptions and limitations. A No or NA answer to this question will not be perceived well by the reviewers. 
    \item The claims made should match theoretical and experimental results, and reflect how much the results can be expected to generalize to other settings. 
    \item It is fine to include aspirational goals as motivation as long as it is clear that these goals are not attained by the paper. 
  \end{itemize}
  
  \item {\bf Limitations}
  \item[] Question: Does the paper discuss the limitations of the work performed by the authors?
  \item[] Answer: \answerYes{} 
  \item[] Justification: We have discussed assumptions we made in this paper
  (e.g., regarding the formal model of transformers) and what consequences
  they have (see Limitations in Concluding Remarks).
  \item[] Guidelines:
  \begin{itemize}
    \item The answer NA means that the paper has no limitation while the answer No means that the paper has limitations, but those are not discussed in the paper. 
    \item The authors are encouraged to create a separate "Limitations" section in their paper.
    \item The paper should point out any strong assumptions and how robust the results are to violations of these assumptions (e.g., independence assumptions, noiseless settings, model well-specification, asymptotic approximations only holding locally). The authors should reflect on how these assumptions might be violated in practice and what the implications would be.
    \item The authors should reflect on the scope of the claims made, e.g., if the approach was only tested on a few datasets or with a few runs. In general, empirical results often depend on implicit assumptions, which should be articulated.
    \item The authors should reflect on the factors that influence the performance of the approach. For example, a facial recognition algorithm may perform poorly when image resolution is low or images are taken in low lighting. Or a speech-to-text system might not be used reliably to provide closed captions for online lectures because it fails to handle technical jargon.
    \item The authors should discuss the computational efficiency of the proposed algorithms and how they scale with dataset size.
    \item If applicable, the authors should discuss possible limitations of their approach to address problems of privacy and fairness.
    \item While the authors might fear that complete honesty about limitations might be used by reviewers as grounds for rejection, a worse outcome might be that reviewers discover limitations that aren't acknowledged in the paper. The authors should use their best judgment and recognize that individual actions in favor of transparency play an important role in developing norms that preserve the integrity of the community. Reviewers will be specifically instructed to not penalize honesty concerning limitations.
  \end{itemize}
  
  \item {\bf Theory Assumptions and Proofs}
  \item[] Question: For each theoretical result, does the paper provide the full set of assumptions and a complete (and correct) proof?
  \item[] Answer: \answerYes{} 
  \item[] Justification: Since this is a theory paper, we try to be explicit in
  all the assumptions that we made for our results. We have also included
  all the proofs for the claims we made in the paper.
  \item[] Guidelines:
  \begin{itemize}
    \item The answer NA means that the paper does not include theoretical results. 
    \item All the theorems, formulas, and proofs in the paper should be numbered and cross-referenced.
    \item All assumptions should be clearly stated or referenced in the statement of any theorems.
    \item The proofs can either appear in the main paper or the supplemental material, but if they appear in the supplemental material, the authors are encouraged to provide a short proof sketch to provide intuition. 
    \item Inversely, any informal proof provided in the core of the paper should be complemented by formal proofs provided in appendix or supplemental material.
    \item Theorems and Lemmas that the proof relies upon should be properly referenced. 
  \end{itemize}
  
  \item {\bf Experimental Result Reproducibility}
  \item[] Question: Does the paper fully disclose all the information needed to reproduce the main experimental results of the paper to the extent that it affects the main claims and/or conclusions of the paper (regardless of whether the code and data are provided or not)?
  \item[] Answer: \answerNA{} 
  \item[] Justification: The paper contains no experiments.
  \item[] Guidelines:
  \begin{itemize}
    \item The answer NA means that the paper does not include experiments.
    \item If the paper includes experiments, a No answer to this question will not be perceived well by the reviewers: Making the paper reproducible is important, regardless of whether the code and data are provided or not.
    \item If the contribution is a dataset and/or model, the authors should describe the steps taken to make their results reproducible or verifiable. 
    \item Depending on the contribution, reproducibility can be accomplished in various ways. For example, if the contribution is a novel architecture, describing the architecture fully might suffice, or if the contribution is a specific model and empirical evaluation, it may be necessary to either make it possible for others to replicate the model with the same dataset, or provide access to the model. In general. releasing code and data is often one good way to accomplish this, but reproducibility can also be provided via detailed instructions for how to replicate the results, access to a hosted model (e.g., in the case of a large language model), releasing of a model checkpoint, or other means that are appropriate to the research performed.
    \item While NeurIPS does not require releasing code, the conference does require all submissions to provide some reasonable avenue for reproducibility, which may depend on the nature of the contribution. For example
    \begin{enumerate}
      \item If the contribution is primarily a new algorithm, the paper should make it clear how to reproduce that algorithm.
      \item If the contribution is primarily a new model architecture, the paper should describe the architecture clearly and fully.
      \item If the contribution is a new model (e.g., a large language model), then there should either be a way to access this model for reproducing the results or a way to reproduce the model (e.g., with an open-source dataset or instructions for how to construct the dataset).
      \item We recognize that reproducibility may be tricky in some cases, in which case authors are welcome to describe the particular way they provide for reproducibility. In the case of closed-source models, it may be that access to the model is limited in some way (e.g., to registered users), but it should be possible for other researchers to have some path to reproducing or verifying the results.
    \end{enumerate}
  \end{itemize}

  \item {\bf Open access to data and code}
  \item[] Question: Does the paper provide open access to the data and code, with sufficient instructions to faithfully reproduce the main experimental results, as described in supplemental material?
  \item[] Answer: \answerNA{} 
  \item[] Justification: The paper contains no experiments and no code.
  \item[] Guidelines:
  \begin{itemize}
    \item The answer NA means that paper does not include experiments requiring code.
    \item Please see the NeurIPS code and data submission guidelines (\url{https://nips.cc/public/guides/CodeSubmissionPolicy}) for more details.
    \item While we encourage the release of code and data, we understand that this might not be possible, so “No” is an acceptable answer. Papers cannot be rejected simply for not including code, unless this is central to the contribution (e.g., for a new open-source benchmark).
    \item The instructions should contain the exact command and environment needed to run to reproduce the results. See the NeurIPS code and data submission guidelines (\url{https://nips.cc/public/guides/CodeSubmissionPolicy}) for more details.
    \item The authors should provide instructions on data access and preparation, including how to access the raw data, preprocessed data, intermediate data, and generated data, etc.
    \item The authors should provide scripts to reproduce all experimental results for the new proposed method and baselines. If only a subset of experiments are reproducible, they should state which ones are omitted from the script and why.
    \item At submission time, to preserve anonymity, the authors should release anonymized versions (if applicable).
    \item Providing as much information as possible in supplemental material (appended to the paper) is recommended, but including URLs to data and code is permitted.
  \end{itemize}

  \item {\bf Experimental Setting/Details}
  \item[] Question: Does the paper specify all the training and test details (e.g., data splits, hyperparameters, how they were chosen, type of optimizer, etc.) necessary to understand the results?
  \item[] Answer: \answerNA{} 
  \item[] Justification: The paper contains no experiments.
  \item[] Guidelines:
  \begin{itemize}
    \item The answer NA means that the paper does not include experiments.
    \item The experimental setting should be presented in the core of the paper to a level of detail that is necessary to appreciate the results and make sense of them.
    \item The full details can be provided either with the code, in appendix, or as supplemental material.
  \end{itemize}
  
  \item {\bf Experiment Statistical Significance}
  \item[] Question: Does the paper report error bars suitably and correctly defined or other appropriate information about the statistical significance of the experiments?
  \item[] Answer: \answerNA{} 
  \item[] Justification: The paper contains no experiments.
  \item[] Guidelines:
  \begin{itemize}
    \item The answer NA means that the paper does not include experiments.
    \item The authors should answer "Yes" if the results are accompanied by error bars, confidence intervals, or statistical significance tests, at least for the experiments that support the main claims of the paper.
    \item The factors of variability that the error bars are capturing should be clearly stated (for example, train/test split, initialization, random drawing of some parameter, or overall run with given experimental conditions).
    \item The method for calculating the error bars should be explained (closed form formula, call to a library function, bootstrap, etc.)
    \item The assumptions made should be given (e.g., Normally distributed errors).
    \item It should be clear whether the error bar is the standard deviation or the standard error of the mean.
    \item It is OK to report 1-sigma error bars, but one should state it. The authors should preferably report a 2-sigma error bar than state that they have a 96\% CI, if the hypothesis of Normality of errors is not verified.
    \item For asymmetric distributions, the authors should be careful not to show in tables or figures symmetric error bars that would yield results that are out of range (e.g. negative error rates).
    \item If error bars are reported in tables or plots, The authors should explain in the text how they were calculated and reference the corresponding figures or tables in the text.
  \end{itemize}
  
  \item {\bf Experiments Compute Resources}
  \item[] Question: For each experiment, does the paper provide sufficient information on the computer resources (type of compute workers, memory, time of execution) needed to reproduce the experiments?
  \item[] Answer: \answerNA{} 
  
  \item[] Justification: The paper contains no experiments.
  \item[] Guidelines:
  \begin{itemize}
    \item The answer NA means that the paper does not include experiments.
    \item The paper should indicate the type of compute workers CPU or GPU, internal cluster, or cloud provider, including relevant memory and storage.
    \item The paper should provide the amount of compute required for each of the individual experimental runs as well as estimate the total compute. 
    \item The paper should disclose whether the full research project required more compute than the experiments reported in the paper (e.g., preliminary or failed experiments that didn't make it into the paper). 
  \end{itemize}
  
  \item {\bf Code Of Ethics}
  \item[] Question: Does the research conducted in the paper conform, in every respect, with the NeurIPS Code of Ethics \url{https://neurips.cc/public/EthicsGuidelines}?
  \item[] Answer: \answerYes{} 
  \item[] Justification: This is a theory paper with no reference to human
  subjects and datasets. 
  Our
  paper does not have neither direct societal impact nor potential harmful
  consequences. We have also carefully checked the impact mitigation
  measures, which do not directly apply to our paper, which
  involves neither experiments nor code.
  \item[] Guidelines:
  \begin{itemize}
    \item The answer NA means that the authors have not reviewed the NeurIPS Code of Ethics.
    \item If the authors answer No, they should explain the special circumstances that require a deviation from the Code of Ethics.
    \item The authors should make sure to preserve anonymity (e.g., if there is a special consideration due to laws or regulations in their jurisdiction).
  \end{itemize}

  \item {\bf Broader Impacts}
  \item[] Question: Does the paper discuss both potential positive societal impacts and negative societal impacts of the work performed?
  \item[] Answer: \answerNA{}
  \item[] Justification: This is a theory paper with no direct societal impacts.
  \item[] Guidelines:
  \begin{itemize}
    \item The answer NA means that there is no societal impact of the work performed.
    \item If the authors answer NA or No, they should explain why their work has no societal impact or why the paper does not address societal impact.
    \item Examples of negative societal impacts include potential malicious or unintended uses (e.g., disinformation, generating fake profiles, surveillance), fairness considerations (e.g., deployment of technologies that could make decisions that unfairly impact specific groups), privacy considerations, and security considerations.
    \item The conference expects that many papers will be foundational research and not tied to particular applications, let alone deployments. However, if there is a direct path to any negative applications, the authors should point it out. For example, it is legitimate to point out that an improvement in the quality of generative models could be used to generate deepfakes for disinformation. On the other hand, it is not needed to point out that a generic algorithm for optimizing neural networks could enable people to train models that generate Deepfakes faster.
    \item The authors should consider possible harms that could arise when the technology is being used as intended and functioning correctly, harms that could arise when the technology is being used as intended but gives incorrect results, and harms following from (intentional or unintentional) misuse of the technology.
    \item If there are negative societal impacts, the authors could also discuss possible mitigation strategies (e.g., gated release of models, providing defenses in addition to attacks, mechanisms for monitoring misuse, mechanisms to monitor how a system learns from feedback over time, improving the efficiency and accessibility of ML).
  \end{itemize}
  
  \item {\bf Safeguards}
  \item[] Question: Does the paper describe safeguards that have been put in place for responsible release of data or models that have a high risk for misuse (e.g., pretrained language models, image generators, or scraped datasets)?
  \item[] Answer: \answerNA{} 
  \item[] Justification: This is a theory paper with neither experiments, codes,
  nor the use of datasets.
  \item[] Guidelines:
  \begin{itemize}
    \item The answer NA means that the paper poses no such risks.
    \item Released models that have a high risk for misuse or dual-use should be released with necessary safeguards to allow for controlled use of the model, for example by requiring that users adhere to usage guidelines or restrictions to access the model or implementing safety filters. 
    \item Datasets that have been scraped from the Internet could pose safety risks. The authors should describe how they avoided releasing unsafe images.
    \item We recognize that providing effective safeguards is challenging, and many papers do not require this, but we encourage authors to take this into account and make a best faith effort.
  \end{itemize}
  
  \item {\bf Licenses for existing assets}
  \item[] Question: Are the creators or original owners of assets (e.g., code, data, models), used in the paper, properly credited and are the license and terms of use explicitly mentioned and properly respected?
  \item[] Answer: \answerNA{} 
  \item[] Justification: This is a theory paper that uses no code, data, and 
  models. We do not use assets of others in any form.
  \item[] Guidelines:
  \begin{itemize}
    \item The answer NA means that the paper does not use existing assets.
    \item The authors should cite the original paper that produced the code package or dataset.
    \item The authors should state which version of the asset is used and, if possible, include a URL.
    \item The name of the license (e.g., CC-BY 4.0) should be included for each asset.
    \item For scraped data from a particular source (e.g., website), the copyright and terms of service of that source should be provided.
    \item If assets are released, the license, copyright information, and terms of use in the package should be provided. For popular datasets, \url{paperswithcode.com/datasets} has curated licenses for some datasets. Their licensing guide can help determine the license of a dataset.
    \item For existing datasets that are re-packaged, both the original license and the license of the derived asset (if it has changed) should be provided.
    \item If this information is not available online, the authors are encouraged to reach out to the asset's creators.
  \end{itemize}
  
  \item {\bf New Assets}
  \item[] Question: Are new assets introduced in the paper well documented and is the documentation provided alongside the assets?
  \item[] Answer: \answerNA{} 
  \item[] Justification: We do not release new assets. 
  \item[] Guidelines:
  \begin{itemize}
    \item The answer NA means that the paper does not release new assets.
    \item Researchers should communicate the details of the dataset/code/model as part of their submissions via structured templates. This includes details about training, license, limitations, etc. 
    \item The paper should discuss whether and how consent was obtained from people whose asset is used.
    \item At submission time, remember to anonymize your assets (if applicable). You can either create an anonymized URL or include an anonymized zip file.
  \end{itemize}
  
  \item {\bf Crowdsourcing and Research with Human Subjects}
  \item[] Question: For crowdsourcing experiments and research with human subjects, does the paper include the full text of instructions given to participants and screenshots, if applicable, as well as details about compensation (if any)? 
  \item[] Answer: \answerNA{} 
  \item[] Justification: This is a theory paper with no crowdsourcing and no
  human subjects. 
  \item[] Guidelines:
  \begin{itemize}
    \item The answer NA means that the paper does not involve crowdsourcing nor research with human subjects.
    \item Including this information in the supplemental material is fine, but if the main contribution of the paper involves human subjects, then as much detail as possible should be included in the main paper. 
    \item According to the NeurIPS Code of Ethics, workers involved in data collection, curation, or other labor should be paid at least the minimum wage in the country of the data collector. 
  \end{itemize}
  
  \item {\bf Institutional Review Board (IRB) Approvals or Equivalent for Research with Human Subjects}
  \item[] Question: Does the paper describe potential risks incurred by study participants, whether such risks were disclosed to the subjects, and whether Institutional Review Board (IRB) approvals (or an equivalent approval/review based on the requirements of your country or institution) were obtained?
  \item[] Answer: \answerNA{} 
  \item[] Justification: This is a theory paper with no crowdsourcing and no
  human subjects. 
  \item[] Guidelines:
  \begin{itemize}
    \item The answer NA means that the paper does not involve crowdsourcing nor research with human subjects.
    \item Depending on the country in which research is conducted, IRB approval (or equivalent) may be required for any human subjects research. If you obtained IRB approval, you should clearly state this in the paper. 
    \item We recognize that the procedures for this may vary significantly between institutions and locations, and we expect authors to adhere to the NeurIPS Code of Ethics and the guidelines for their institution. 
    \item For initial submissions, do not include any information that would break anonymity (if applicable), such as the institution conducting the review.
  \end{itemize}
  
\end{enumerate}
}

\ifConference{%

}

\label{afterbibliography}
\newoutputstream{pagestotal}
\openoutputfile{main.pagestotal.ctr}{pagestotal}
\addtostream{pagestotal}{\getpagerefnumber{afterbibliography}}
\closeoutputstream{pagestotal}

\newoutputstream{todos}
\openoutputfile{main.todos.ctr}{todos}
\addtostream{todos}{\arabic{@todonotes@numberoftodonotes}}
\closeoutputstream{todos}

\end{document}